\documentclass[submission,copyright,creativecommons]{eptcs}
\providecommand{\event}{GCM 2020} 
\def\titlerunning{Graph repair and application to meta-modeling}
\def\authorrunning{Christian Sandmann} 

\usepackage[utf8]{inputenc}
\usepackage{graphicx}
\usepackage{amsmath}
\usepackage{amssymb}
\usepackage{amsfonts}
\usepackage{extarrows}
\usepackage{amstext}
\usepackage{stmaryrd}
\usepackage{latexsym}
\usepackage{rotating}
\usepackage{color}
\usepackage{fancybox}
\usepackage{tikz}
\usetikzlibrary{chains,fit,positioning,petri}
\usepackage{url}
\usepackage{dsfont} 

\usepackage{cutwin}
\usepackage{algorithm2e}
\definecolor{skyblue}{RGB}{0,255,255}
\tikzstyle{labcolor}=[fill=skyblue]  
\usepackage{styles/tikz-atposition}
\usepackage{styles/tikz-circlesnarrows1}
\usepackage{styles/tikz-colorsnstyles}
\usepackage{styles/tikz-extension}
\usepackage{styles/tikz-derivation}
\usepackage{styles/tikz-vcond} 

\usepackage{wrapfig}
\usepackage{styles/tikz-lok}
\usepackage{styles/tikz-railroad}

\providecommand\ignore[1]{}

\usepackage{styles/mytheorem}
\usepackage{styles/abbr-habel} 
\usepackage{styles/abbr-habel-tikz}
\usepackage{styles/abbr-habel-tikz2} 
\usepackage{styles/abbr-poskitt-plump}
\usepackage{styles/abbr-hubatschek-tikz}
\usepackage{styles/abbr-radke}
\usepackage{styles/abbr-sandmann-tikz}
\usepackage{styles/tikz-railroad-macro}
\usepackage{styles/EMF-style}
\usetikzlibrary{positioning,decorations.pathreplacing,arrows.meta}
\usetikzlibrary{arrows,automata,calc}
\usepackage{tabularx}
\usepackage{microtype}
\usepackage{etoolbox}
\apptocmd{\sloppy}{\hbadness 10000\relax}{}{}

\tikzstyle{round corners}=[rounded corners=0.5ex]
\usepackage{float}
\usepackage{graphics}
\usepackage{booktabs}
\providecommand{\shortv}[1]{#1}
\providecommand{\longv}[1]{} 
\providecommand{\examplev}[1]{#1} 
\providecommand{\ap}[1]{} 
\providecommand{\forget}[1]{} 
\renewcommand{\red}{\color{black}}		
 
\providecommand{\CS}[1]{{\red CS:  #1}} 
\newcommand{\legit}{{legit} }

\newcommand{\emfk}{\mathrm{emf}k}
\newcommand{\emf}{\mathrm{emf}}
\newcommand{\vmod}{\rotatebox[origin=c]{-180}{$\models$}}

\begin{document}
\title{Graph Repair and its Application to \\Meta-Modeling\thanks{This work is partly supported by the German Research Foundation (DFG), Grants HA 2936/4-2 and TA 2941/3-2 (Meta-Modeling and Graph Grammars: Generating Development Environments for Modeling Languages).}
}
\author{Christian Sandmann\institute{Universit\"at Oldenburg\\ Oldenburg, Germany \email{christian.sandmann@uni-oldenburg.de}}}
\maketitle 



\begin{abstract} 
Model repair is an essential topic in model-driven engineering.
We present typed graph-repair programs for specific conditions; application to any typed graph yields a typed graph satisfying the condition. 
A~model graph based on the Eclipse Modeling Framework (EMF), short EMF-model graph, is a typed graph satisfying some structural EMF-constraints. Application of the results to the EMF-world yields model-repair programs for \EMF constraints, a first-order variant of EMF constraints; application to any typed graph yields an \EMF model graph. 
From these results, we derive results for EMF model repair.
\end{abstract}



\section{Introduction}
In model-driven software engineering, the primary artifacts are models\ignore{ \cite{Ehrig-etal15a},} \cite{Sendall-etal03a,Heckel-Taentzer20a}. Models have to be consistent w.r.t. a set of constraints, specified for example in the Object Constraint Language (OCL) \cite{OCL}.
To increase the productivity of software development, it is necessary to automatically detect and resolve inconsistencies arising during the development process \ignore{called model repair }(see, e.g. \cite{Nentwich-etal03a,Macedo-etal17a,Nebras-etal17a}). 


To enable automated model repair or model completion, we look for an algorithm that - given a meta-model with two constraints and any model satisfying one of the constraints - creates another model satisfying the old as well as a new one (see Figure \ref{fig:introduction}).  

\begin{figure}[h]
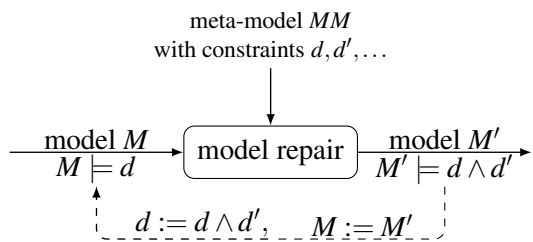

\[\scalebox{1}{
\tikz[node distance=6em,rounded corners]{
\node(input) {};
\node[strictly right of=input,inner sep=5pt,draw] (repair) {model repair};
\node[strictly right of=repair] (output) {};
\node(h)[node distance=3em,below of=repair]{};

\draw[arrow] (input) to node [above] {model $M$} node [below](a) {$M \models d$} (repair);
\draw[arrow] (repair) to node [above] {model $M'$} node [below](b) {$M' \models d \wedge d'$} (output);
\draw[dashed,-] (b) |- node[above left] {$M : = M'$~~~~~} (h);
\draw[dashed,arrow] (h) -| node[above right] {~~~~~$d := d \wedge d',$} (a);

\node (MM) [node distance=4em,above of=repair] {\footnotesize \begin{tabular}{c}
meta-model $MM$ \\ with constraints $d,d',\ldots$\end{tabular}};
\draw[arrow] (MM) to (repair);
}}\]
\caption{\label{fig:introduction} General idea to model repair}
\end{figure}

If we have such an algorithm, the process can be iterated: Using the model satisfying two constraints, and a new constraint as input, the algorithm creates a model satisfying the conjunction of three constraints, and so on.
This iterative approach is necessary in handling large conditions. In each step, one condition is handled. If all steps terminate, and in all steps the preceding conditions remain preserved, we can be sure that, after the consideration of all finitely many conditions, the conjunction of the conditions is satisfied. If after one step a preceding condition is violated, the condition must be considered again and the process may be come non-terminating.

In this paper, we represent a meta-model as a type graph, the instance model as a graph typed over the type graph, and first-order constraints as a typed graph constraint, equivalent to a first-order graph formula \cite{Radke+18a}. Given a typed constraint, we extract a typed program from the constraint, called \emph{repair program}, such that the application of the typed repair program to an arbitrary typed graph yields a typed graph satisfying the constraint.

For small (basic) constraints, we extract a basic repair program directly from the constraint. For larger (proper) constraints, the repair program is composed from basic repair programs. For generalized proper constraints repair programs for subconditions may be used for the construction. For conjunctive constraints we take repair programs for the components and compose them to a typed repair program, provided that there is a sequentialization that preserves the preceding constraints. For disjunctive constraints we need a repair program for one component. Altogether, for so-called legit constraints, we can construct a repair program. These constructions are done in the $\M$-adhesive category of typed graphs with $\Epi'$-$\M$-pair factorization.

\ignore{
For a so-called proper constraint, i.e. without conjunctions and disjunctions, with alternating quantifiers ending with $\ctrue$ or of the form $\PE(a, \NE b)$, we extract a repair program, directly from the constraint. 
For conjunctions of constraints, for which there exist a repair program, we take a sequence of the constraints, construct the repair programs and constraint-preserving versions of them. In more detail, we construct a sequence of repair programs where the first one preserves $\ctrue$, the second one preserves the first constraint, the third one preserves the first two constraints, and so on, and take the sequential composition of these repair programs. }


A model graph based on the Eclipse Modeling Framework (EMF) \cite{Steinmann-etal08a}, short EMF model graph, is a typed graph satisfying some structural EMF constraints. Application of the results for typed graphs to the EMF world yields model completion programs for \EMF constraints, a first-order version of the EMF constraints, such that the application to a typed graph yields an \EMF model graph. The results known from typed graph repair are applied to \EMF model repair and EMF model repair.
 
The structure of the paper is as follows. 
In Section~\ref{sec:preliminaries}, we review the definitions of typed graphs, typed graph conditions, and typed graph programs.
In Section~\ref{sec:repair}, we introduce the concept of typed repair programs and show that there are repair programs for a large number of conditions, so-called legit conditions. Application of a typed repair program to any typed graph yields a typed graph satisfying the constraint.
In Section~\ref{sec:application}, application of the results to EMF-world yields model-repair and completion programs for \EMF constraints, a first-order variant of EMF constraints.  From these results, we derive results of EMF-model repair and completion.
In Section~\ref{sec:related}, we present some related concepts.
In Section~\ref{sec:conclusion}, we give a conclusion and mention some further work.
\section{Preliminaries}\label{sec:preliminaries}

In the following, we recall the definitions of typed graphs, graph conditions, rules and transformations, graph programs, and basic transformations \cite{Biermann-Ermel-Taetzer12a,Habel-Pennemann09a,Pennemann09a}. 
In the following, our concepts are based on \cite{Biermann-Ermel-Taetzer12a}. For simplicity, we ignore the attributes.


A directed graph consists of a set of nodes and a set of edges where each edge is equipped with a source and a target node.

\begin{definition}[graphs \& morphisms]\label{def:graph}
A \emph{(directed) graph} $G=(V_G,E_G,\sou_G,\tar_G)$ consists of a\ignore{finite} set $V_G$ of \emph{nodes} and a\ignore{finite} set $E_G$ of \emph{edges}, as well as source and target functions $\sou_G,\tar_G\colon E_G\to V_G$.
Given graphs $G$ and $H$, a \emph{(graph) morphism} $g\colon G \to H$ consists of  total functions $g_V\colon V_G\to V_H$ and $g_E\colon E_G\to E_H$ that preserve sources and targets, that is, $g_V\circ\sou_G=\sou_H\circ g_E$ and $g_V\circ\tar_G=\tar_H\circ g_E$. The morphism $g$ is \emph{injective}\/ (\emph{surjective}\/) if $g_{\V}$ and $g_{\E}$ are injective (surjective), and an \emph{isomorphism}\/ if it is injective and surjective. In the latter case, $G$ and $H$ are \emph{isomorphic}, denoted by $G\cong H$. \end{definition}

\newpage 
\begin{convention}
Drawing a graph, nodes are drawn as circles  and edges as arrows. 
Arbitrary morphisms are drawn by usual arrows $\to$, injective ones by $\injto$. 
\end{convention}

A type graph (with containment) is a graph with a distinguished set of containment edges, and a relation of opposite edges. 


\begin{definition}[Type graph]\label{def:type-graph}
A~\emph{type graph} $TG = (T,C,O)$ consists of a graph $T$, a set $C\subseteq E_{T}$ of \emph{containment edges}, and a relation $O\subseteq E_{T}\times E_{T}$ of \emph{opposite edges}. The relation $O$ is anti-reflexive, symmetric, functional, i.e., $\forall (e_1,e_2),(e_1,e_3)\in O$, $e_2 = e_3$, and opposite direction, i.e., $\forall (e_1, e_2) \in O$, $\src(e_1)= \tgt(e_2)$ and $\src(e_2)=\tgt(e_1)$.
\end{definition}


\begin{convention}
The drawing of a type graph is obtained from the underlying graph by marking every containment edge (\containment) with a black diamond at the source, and adding, for every pair $(e_1,e_2)$ of opposite edges, a bidirectional edge (\opposite)  between the source and the target of the first edge with two edge type names, one at each end. 
\end{convention}

\ignore{\red The opposite edges, refer to the reference representing the opposite direction of a bidirectional association. Thus, such an association is represented by the two reference edges, each defining the other as its opposite. In UML, containment is a stronger type of association that implies a whole-part relation ship: an object cannot, directly or undirectly, contain its own container, it can have no more than one container, and its life span ends with that of its container.}

\examplev{\begin{example}\label{exa:tgic}
A type graph for Petri-nets \ignore{\cite{Reisig82a} }is given in Figure \ref{fig:mm}.
\begin{figure}[h]
\[\scalebox{0.9}{$
\begin{tikzpicture}[node distance=4.5em]
\node[class] (net) {PetriNet};
\node[class,strictly below right of=net,minimum height=2em] (ptarc) {PTArc};
\node[class,strictly below left of=net,minimum height=2em,yshift=-1em] (tparc) {TPArc};
\node[class,strictly right of=ptarc,minimum height=2em,yshift=-1em] (place) {Place};
\node[class,strictly left of=tparc,minimum height=2em,yshift=1em] (transition) {Transition};
\node[class,strictly below of=place,minimum height=2em] (token) {Token};
\node(h)[strictly above of=tparc,node distance=2em,inner sep=0pt,outer sep=0pt]{};

\draw[containment] (net) -| node[above]{\footnotesize $\mathrm{trans}$} (transition);
\draw[containment] (net) -| node[above]{\footnotesize $\mathrm{places}$} (place);
\draw[containment] (place) edge node[right]{\footnotesize $\token$} (token);
\draw[opposite] ([yshift=0em]ptarc.east) to node[below left]{\footnotesize $\mathrm{out}$} node[above right]{\footnotesize $\mathrm{src}$} ([yshift=1em]place.west);
\draw[opposite] ([yshift=-1em]tparc.east) to node[below left]{\footnotesize $\mathrm{in}$} node[above right]{\footnotesize $\mathrm{tgt}$} ([yshift=-1em]place.west);
\draw[opposite] ([yshift=0em]tparc.west) to node[below left]{\footnotesize $\mathrm{src}$} node[above right]{\footnotesize $\mathrm{out}$} ([yshift=-1em]transition.east);
\draw[opposite] ([yshift=1em]ptarc.west) to node[below left]{\footnotesize $\mathrm{tgt}$} node[above right]{\footnotesize $\mathrm{in}$} ([yshift=1em]transition.east);

\draw[containment] (net) to node[above,right]{\footnotesize $\mathrm{ptarcs}$} (ptarc);
\draw[containmentundir] (net) to node[above,left]{\footnotesize $\mathrm{tparcs}$} (h);
\draw[->] (h) to (tparc);
\end{tikzpicture}
$}\]
\caption{\label{fig:mm}Type graph for Petri-nets, adapted from \cite{Wachsmuth07a}}
\end{figure}
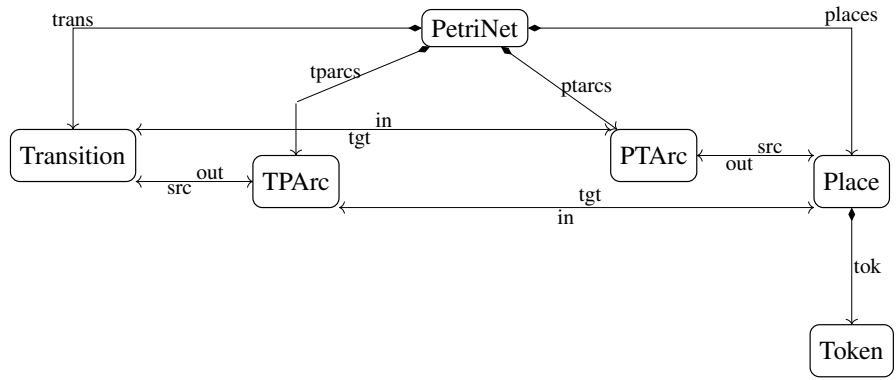

The type graph consists of the nodes PetriNet (PN), Place (Pl), Transition (Tr), Token (Tk), place-to-transition arcs (PTArc), and transition-to-place arcs (TPArc), written inside the nodes, and the edges places, trans, and $\token$. The distinguished containment edges from the PetriNet to the Place (Transition, PTArc, and TPArc)-node are marked in the graph. The opposite edge relation relates the edges from the PTArc (TPArc)-node to the Place (Transition)-node of type src and the Place (Transition)-node to the PTArc (TPArc)-node of type out.
\end{example}}

\begin{assumption}
In the following, let $TG = (T,C,O)$ be a fixed type graph.
\end{assumption}

A typed graph over a type graph is a graph together with a typing morphism. The typing itself is done by a graph morphism between the graph and the type graph.

\begin{definition}[typed graphs]\label{def:typed-graph}
A~\emph{typed graph} $(G,type)$ is a graph $G$ together with a typing morphism $type\colon G \to T$. 
Given typed graphs $(G,type_G)$, $(H,type_H)$, a \emph{typed graph morphism} $g \colon G \to H$ is a graph morphism such that $type_{H,V}\circ g_V=type_{G,V}$ and $type_{H,E}\circ g_E=type_{G,E}$.
For a node $v$ (an edge $e$) in $G$, $type_V(v)$ ($type_E(e)$) is the \emph{node (edge) type}.
\end{definition}

\longv{
\[\tikz[node distance=1.5em,shape=rectangle,outer sep=1pt,inner sep=2pt]{
\node(P){$G$};
\node(G)[strictly below right of=P]{$TG$};
\node(C)[strictly above right of=G]{$H$};
\draw[morphism] (P) -- node[overlay,above](a){$g$} (C);
\draw[morphism] (P) -- node[overlay,below left]{$type_G$} (G);
\draw[morphism] (C) -- node[overlay,below right](q){$type_H$} (G);
\draw[draw=white] (a) -- node[overlay](tr1){=} (G);
}\]
}

\begin{convention}
Given a typed graph $(G,type)$, we draw the graph $G$ and put in type information: For a node $v$ in $G$, we depict the node type $type_V(v)$ inside the node; for an edge $e$ in $G$, we depict the edge type $type_E(e)$ near the target node of the edge $e$. Each edge with edge type containment edge is marked as a containment edge.
For every pair of nodes whose type nodes are connected by an opposite edge, an opposite edge is added.
For each pair $(v_1,v_2)$ of nodes in $G$ for which $type_V(v_1),type_V(v_2)$ are connected by an opposite edge, a bidirectional edge between the nodes with two edge type names, one at each end, is added.
\end{convention}


\begin{assumption}
In the following, all graphs are typed over $TG$ and all morphisms are injective. 
\end{assumption}

{
{ \begin{note}Typed graphs (over $TG$) with containment and morphisms form a category \textbf{Graphs}$_{\textbf{TG}}$. 
This is $\M$-adhesive  and has a $\Epi'$-$\M$ pair factorization \cite{Ehrig-Golas-Hermann10a,Ehrig-Ehrig-Prange-Taentzer06b} where $\M$ is the class of injective\ignore{ and type-preserving} morphisms and $\Epi'$ is the class of pairs of jointly surjective morphisms. $\M$-adhesiveness implies the existence of pushouts\ignore{ along
$\M$-morphisms} (used in Definition~\ref{def:rule} and Lemma~\ref{lem:left}); $\Epi'$-$\M$ pair factorization is used in the shift construction in Lemma~\ref{lem:shift}.\end{note}
}}

\ignore{
\begin{remark}
Typed graphs over a type graph $TG$ and injective morphisms form a category \textbf{Graphs}$_{\textbf{TG}}$. This category is $\M$-adhesive \cite{Ehrig-Golas-Hermann10a} and has a $\Epi'$-$\M$ pair factorization \cite{Ehrig-Ehrig-Prange-Taentzer06a} where $\M$ is the class of injective morphisms, and $\Epi'$ is the class of pairs of jointly surjective morphisms.
By $\M$-adhesiveness, there are pushouts along ``$\M$-morphisms'', used in Definition \ref{def:rule}.
The $\Epi'-\M$-factorization is used in the construction of $\Shift$ in Lemma~\ref{lem:shift}.
\ignore{
{\red The chosen typed graphs are $\M$-adhesive: there exists a mapping from the containment edges of the typed graphs here to the edges of the typed graph, enriched with a special {\red marking} ~~\begin{rotate}{90}$\blacklozenge$\end{rotate}~~, drawn at the source of the containment edge, and the opposite edges is a special relation.}
Consequently, the typed graphs chosen here, form an $\M$-adhesive category.}
\end{remark}
}


Typed graph conditions are nested constructs, which can be represented as trees of morphisms equipped with quantifiers and Boolean connectives. Graph conditions and first-order graph formulas are expressively equivalent.

\begin{definition}[typed graph conditions]\label{def:cond}
A \emph{(typed graph) condition} over a graph $A$ is of the form (a) $\ctrue$ or $\exists(a,c)$ where $a\colon A \injto C$ is a real inclusion morphism, i.e., $A \subset C$, and $c$ is a condition over $C$. (b) For a condition $c$ over~$A$, $\neg c$ is a condition over~$A$.
(c) For conditions $c_i$ ($i \in I$ for some finite index set $I$\longv{\footnote{In this paper, we consider graph conditions with finite index sets.}}) over $A$, $\wedge_{i \in I} c_i$ is a condition over~$A$. \ignore{Conditions built by (a) and (b) are called \emph{linear}.} 
Conditions over the empty graph~$\emptyset$ are called \emph{constraints}. In the context of rules, conditions are called \emph{application conditions}. 
Any morphism $p\colon A\injto G$ \emph{satisfies} $\ctrue$. A~morphism $p$ \emph{satisfies} $\PE(a,c)$ with $a\colon~A\injto C$ if there exists an morphism $q\colon C\injto G$ such that $q\circ a=p$ and $q$~satisfies~$c$. 
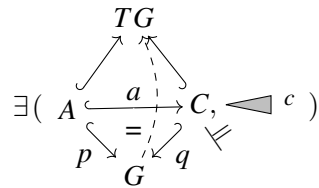
\[\tikz[node distance=1.5em,shape=rectangle,outer sep=1pt,inner sep=2pt]{
\node(P){$A$};
\node(G)[strictly below right of=P]{$G$};
\node(C)[strictly above right of=G]{$C,$};
\draw[monomorphism] (P) -- node[overlay,above](a){$a$} (C);
\draw[monomorphism] (P) -- node[overlay,below left]{$p$} (G);
\draw[altmonomorphism] (C) -- node[overlay,below right](q){$q$} (G);
\draw[draw=white] (a) -- node[overlay](tr1){=} (G);
\node(c)[outer sep=0pt,inner sep=0pt,node distance=0em,strictly right of=C]{\tikz[draw=black,fill=lightgray]{
\filldraw (0,0) -- (0.6,0.12) -- node[right,outer sep=1ex]{\footnotesize{$c$}} (0.6,0) -- (0.6,-0.12) -- (0,0);}};
\draw[draw=white] (q) -- node[overlay,sloped](tr1){$\models$} (c);
\node(Y)[node distance=0.2em,strictly right of=c]{$)$};
\node(X)[node distance=0.0em,strictly left of=P]{$\PE($};
\node(TGI)[strictly above of=a]{$TG$};
\draw[monomorphism] (P) edge node[overlay,above,sloped]{\footnotesize} (TGI);
\draw[monomorphism] (C) edge node[overlay,above,sloped]{\footnotesize} (TGI);
\draw[morphism,dashed] (G) edge [bend right=25] (TGI);
\draw[draw=white] (a) -- node[overlay](tr1){} (TGI);}\]

A~morphism $p$ \emph{satisfies} $\neg c$ if $p$ does not satisfy $c$, and $p$ \emph{satisfies} $\wedge _{i \in I}c_i$ if $p$ satisfies each $c_i$ ($i \in I$). We write  $p\models c$ if $p$ satisfies the condition $c$ (over $A$). A~graph $G$ \emph{satisfies} a constraint $c$, $G\models c$, if the morphism $p\colon\emptyset\injto G$ satisfies~$c$. A constraint $c$ is \emph{satisfiable} if there is a graph $G$ that satisfies $c$. 
\ignore{$\psem{c}$~denotes the class of all graphs satisfying $c$. 
Two conditions $c$ and $c'$ over $A$ are \emph{equivalent}, denoted by $c\equiv c'$, if for all graphs $G$ and all  morphisms $p\colon A\injto G$, $p\models c$ iff $p\models c'$. A condition $c$ \emph{implies} a condition $c'$, denoted by $c\impl c'$, if for all graphs and all  morphisms $p\colon A\injto G$, $p\models c$ implies $p\models c'$.}
\end{definition}
\begin{notation}Conditions may be written in a more compact form: $\PE a:=\PE(a,\ctrue)$, $\cfalse:=\neg \ctrue$, $\PA(a,c):=\NE(a, \neg c)$, and $\NE:=\neg\PE$\ignore{, and $\NA:=\neg\PA$}. The expressions $\vee_{i \in I} c_i$ and $c\impl c'$ are defined as usual. For a morphism $a\colon A\DSinjto C$ in a condition, we just depict the codomain $C$, if the domain $A$ can be unambiguously inferred.\end{notation}

\ignore{\begin{example}
The expression \[\PA (\emptyset \injto \directededgelab{Tr}{TPArc}, \PE \directededgelab{Tr}{TPArc} \injto \oppositeedgelab{Tr}{TPArc}, \ctrue))\] is constraint according to Definition \ref{def:cond}, written in compact form as\\ $\PA (\directededgelab{Tr}{TPArc}, \PE \oppositeedgelab{Tr}{TPArc}))$ meaning that there are no real\footnote{An edge is said to be \emph{real} if it is not a loop.} incoming reference edges to the Petri-net node. The type graph with inheritance and containment is given in Example \ref{exa:tgic}, where the typing morphism can be unambiguously inferred.
\end{example}}

\ignore{\begin{definition}[conditions with alternating quantifiers]Conditions of the form $\Q(A_1,\bar{\Q}(A_2,\Q(A_3,\ldots)))$ with $\Q\in\{\PA,\PE\}$, $\bar\PA=\PE$, $\bar\PE=\PA$ ending with $\ctrue$ or $\cfalse$ are conditions with \emph{alternating quantifiers}. A condition with alternating quantifiers ending with $\ctrue$ is \emph{pure} and \emph{proper} if it is pure or of the form $\PE(A,\NE C)$. 
\end{definition}

\begin{fact}[alternating quantifiers]\label{fac:alternate} For every condition (without conjunctions and disjunctions), an equivalent condition with alternating quantifiers can be constructed.\end{fact}
\shortv{\begin{proof}Given a condition $d$, by a normal form result\longv{ in \cite[Theorem 2]{Pennemann04a}}, an equivalent condition $d'$ in normal form can be constructed. Applying the rule $\NE(a, \neg c)\equiv \PA(a, \PE c)$ as long as possible to $d'$, yields an equivalent condition with alternating quantifiers.\end{proof}}

\longv{Pure conditions with alternating quantifiers are satisfiable, non-pure conditions are not satisfiable ($\cfalse$) or end with a condition equivalent to $\NE C$ ($\PA(C,\false)\equiv\NE C$) and may be satisfiable.}

\begin{fact}Proper conditions are satisfiable.\end{fact}
\shortv{\begin{proof}A proper condition is $\ctrue$, ends with a condition of the form $\PE(x,\ctrue)\equiv\PE x$ or  $\PA(x,\ctrue)\equiv\ctrue$, or is of the form $\PE(A,\NE C)$. Thus, it is satisfiable.\end{proof}}

\ignore{$\PE(x,\ctrue)\equiv\PE x$\\
$\PA(x,\ctrue)\equiv\ctrue$\\
$\PE(x,\cfalse)\equiv\cfalse$\\
$\PA(x,\false)\equiv\NE x$\\
$\PA(x,\PE(y,\cfalse)\equiv\PA(x,\cfalse)\equiv\NE x$\\
$\PE x,\PA(x,\false)\equiv\PE(x,\NE y)$\\}
}


\longv{Plain rules are specified by a pair of injective morphisms. {\red The matches are required to be type-refining.} They may be equipped with application conditions, and interfaces. By the interfaces, it becomes possible to hand over information between the transformation steps. }

The following is done in the framework of $\M$-adhesive categories.
Rules are specified by a pair of morphisms, interface morphisms, and an application condition. By the interfaces, it becomes possible to hand over information between the transformation steps.\ignore{ The ``horizontal'' morphisms are type-preserving, the ``vertical'' ones type-refining.} In contrast to \cite{Taentzer12a}, our vertical morphisms are injective.
\begin{definition}[typed rules \& transformations]\label{def:rule} Given a category $\C$, a \emph{(typed) rule} $\prule=\tuple{x,p, \ac, y}$ (\emph{with interfaces} $X,Y$) consists of a plain rule $p = \tuple{L \injlto K \injto R}$ of morphisms $l \colon K \injto L, r \colon K \injto R$, morphisms $x\colon X\injto L$, $y\colon Y\injto R$, the \emph{(left and right) interface morphisms}, and a left application condition $\ac$ over $L$. The partial morphism $i \colon X \injto Y$ with $i = y^{-1} \circ r \circ l^{-1} \circ x$ is the \emph{interface morphism} of the rule $\prule$.
 
If the domain of an interface morphisms is empty or the application condition $\ac$ is $\ctrue$, then the component may not be written.

A \emph{direct transformation}\/ from $G$ to $H$ applying $\prule$ at $g\colon X\injto G$ consists of the following steps: 

\begin{enumerate}
\item[(1)] Mark a morphism $g'\colon L\injto G$, called \emph{match} satisfying the dangling condition, such that $g=g'\circ x$ and $g'\models \ac$.
\item[(2)] Apply the plain rule $p$ at $g'$\ignore{\red (possibly)} yielding a morphism \ignore{comatch }$h'\colon R\injto H$.
\item[(3)]Unmark $h\colon Y\injto H$, i.e., define $h=h'\circ y$.
\end{enumerate}
\vspace*{-0.5em}
\begin{figure}[h]
\[\scalebox{1}{$\begin{tikzpicture}[node distance=3.5em,shape=rectangle,outer sep=1pt,inner sep=2pt,label distance=-1.25em]
\node(X){$X$};
\node(L)[strictly right of=X]{$L$};
\node(K)[strictly right of=L]{$K$};
\node(R)[strictly right of=K]{$R$};
\node(Y)[strictly right of=R]{$Y$};
\node(D)[strictly below of=K]{$D$};
\node(G)[strictly left of=D]{$G$};
\node(H)[strictly right of=D]{$H$};
\node(TGI)[strictly above of=K]{$TG$};
\draw[monomorphism] (X) -- node[overlay,above]{\footnotesize $x$} (L);
\draw[altmonomorphism] (K) -- node[overlay,above]{\footnotesize $l$} (L);
\draw[monomorphism] (K) -- node[overlay,above]{\footnotesize $r$} (R);
\draw[altmonomorphism] (Y) -- node[overlay,above]{\footnotesize $y$} (R);
\draw[altmonomorphism] (D) -- node[overlay,above]{\footnotesize $l^{*}$}(G);
\draw[monomorphism] (D) -- node[overlay,above]{\footnotesize $r^{*}$}(H);

\draw[monomorphism] (X) -- node[overlay,below left](g){\footnotesize $g$} (G);
\draw[monomorphism] (L) -- node[overlay,left]{\footnotesize $g'$} (G);
\draw[monomorphism] (K) -- node[overlay,left]{\footnotesize } (D);
\draw[monomorphism] (R) -- node[overlay,right]{\footnotesize $h'$}(H);
\draw[monomorphism] (Y) -- node[overlay,below right](h){\footnotesize $h$}(H);
\draw[draw=none] (L) -- node[overlay]{\footnotesize (1)} (D);
\draw[draw=none] (R) -- node[overlay]{\footnotesize (2)} (D);
\draw[draw=white] (g) -- node[overlay]{=} (L);
\draw[draw=white] (h) -- node[overlay]{=} (R);

\draw[morphism,dashed] (X) edge[bend right=14] node[above right]{\footnotesize $i$} (Y);

\draw[morphism] (L) -- node[overlay,above,sloped]{\footnotesize} (TGI);
\draw[morphism] (K) edge node[overlay,above,sloped]{\footnotesize } (TGI);
\draw[morphism] (R) edge node[overlay,above,sloped]{\footnotesize } (TGI);
\draw[morphism] (X) edge [bend right=-20] node[overlay,above]{\footnotesize} (TGI);
\draw[morphism] (Y) edge [bend right=20] node[overlay,above]{\footnotesize} (TGI);
\node(c)[outer sep=0pt,inner sep=0pt,node distance=0em,strictly above of=L]{\tikz[draw=black,fill=lightgray]{
\filldraw (0,0) -- (-0.12,0.5) -- node[above,outer sep=1ex]{\footnotesize{$\ac$}} (0.12,0.5) -- (0,0);}};
\end{tikzpicture}$}\]
\caption{\label{fig:DPO}A direct transformation}
\end{figure}
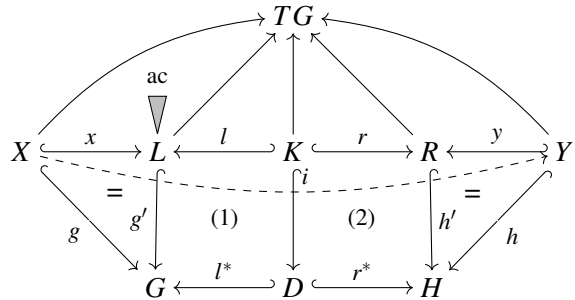
The application of a plain rule is as in the double-pushout approach \cite{Ehrig-Ehrig-Prange-Taentzer06b} in the category of typed graphs. A plain rule $p = \tuple{L \injlto K \injto R}$ is applicable to a graph $G$ w.r.t. a morphism $g' \colon L \injto G$, iff $g'$ satisfies the \emph{dangling condition}: ``No edge in $G - g'(L)$ is incident to a node in $g'(L-K)$.''

The \emph{semantics} of the rule $\prule$ is the set $\psem{\prule}$ of all triples $\tuple{g,h,i}$ of a morphism $g\colon X\injto G$, a morphism $h\colon Y\injto H$, and a partial interface morphism $i\colon X\injto Y$ with $i=y^{-1} \circ r \circ l^{-1}\circ x$. Instead of $\tuple{g,h,i}\in\psem{\prule}$, we write $g\dder_{\prule,i}h$ or short $g\dder_{\prule}h$.\ignore{, $G\dder_{\prule,g,h,i} H$ or short $G\dder_\prule H$. }
\end{definition}

\textbf{The dangling edges operator.} For node-deleting rules~$\prule$, the dangling condition may not be satisfied. 
In this case, $\prule'$ means that the rule shall be applied in the SPO-style of replacement \cite{Loewe93a}, i.e., first to remove the dangling edges, and, afterwards application of the rule in the DPO style. 
Note that this style of replacement also can be described by a DPO-program that fixes a match for the rule, deletes the dangling edges, and afterwards applies the rule at the match. The proceeding can be extended to sets of rules: For a rule set $\S$, $\S'=\{\prule'\mid\prule\in\S\}$.
\vspace{1em}

\textbf{Interface \& markings.} Rules with interfaces enable the control over marking and unmarking of elements in a typed graph and are capable of handling the markings over transformation steps. The left interface restricts the application of the rule to a previously marked context: Given a morphism $g \colon X \injto G$, the application is restricted to those morphisms $g'\colon L \injto G$ that fit to $g$, i.e. $g = g' \circ x$.\ignore{ In the case that the left interface is empty, there is no restriction for the application of the rule.} The right interface restricts the application of the next rule: By the morphism $h \colon Y \injto H$, the next rule can only be applied at~$Y$. 
Instead of rules with interfaces in the sense of \cite{Pennemann09a}, we could use markings as, e.g., in \cite{Habel-Sandmann18a,Poskitt-Plump13a}. \longv{Rules with interfaces may be seen as rules with markings: Whenever there is a marking of $A$ in a graph $G$, i.e., a morphism from $A$ to $G$, choose an extended marking of $C$ in $G$, i.e. a morphism from $C$ to $G$, apply the marked program at that marked position, and, finally, unmark the occurrence. 
Rules with interfaces may be seen as a formal {\red morphism-based} version of the idea of markings combined with rules. }
We have decided to use the interfaces instead of markings, because we have nested markings and the description of markings by morphisms makes transparent what happens.


Typed graph programs are made of sets of typed rules with interface, non-deterministic choice $\{P,Q\}$, sequential composition $\tuple{P;Q}$, the try-statement $\try P$, and the as long as possible iteration $P\downarrow$. 

\begin{definition}[typed graph programs]\label{def:prog}\label{terminating} The set of \emph{(typed graph) programs with interface $X$}, $\Prog(X)$, is defined inductively: 
\[\begin{tabular}{lll}
(1)& Every typed rule $\prule$ with interface $X$ (and $Y$) is in $\Prog(X)$.\\
(2)& If $P,Q\in\Prog(X)$, then $\{P,Q\}$ is in $\Prog(X)$.\\
(3)& If $P\in\Prog(X)$ and $Q\in\Prog(Y)$, then $\tuple{P;Q}\in\Prog(X)$. &\\
(4)& If $P\in\Prog(X)$, then $\try P$, and $P\downarrow$ are in $\Prog(X)$. &\\
\end{tabular}\]
The statement $\Skip$ denotes the identity rule $\id_X = \tuple{X \injlto X \injto X}$.  

The \emph{semantics} of a program $P$ with interface $X$, denoted by $\psem{P}$, is a set of triples such that, for all $\tuple{g,h,i}\in\psem{P}$, the domain of $g$ and $i$ is $X$ and the codomain of $h$ and $i$ is equal:
\[\begin{array}{llcl}
(1) &\psem{\prule}&&\mbox{as in Definition \ref{def:rule}},\\
(2)&\psem{\{P,Q\}} &=& \psem{P}\cup\psem{Q},\\
(3) &\psem{\tuple{P;Q}}&=&\{\tuple{g_1,h_2,i_2{\circ}i_1}\mid \tuple{g_1,h_1,i_1}{\in}\psem{P}, \tuple{g_2,h_2,i_2}{\in}\psem{Q}\mbox{, }h_1=g_2\},\\
(4) &\psem{\try P}& =& \{\tuple{g,h,i} \mid \tuple{g,h,i} \in \psem{P}\} \cup \{\tuple{g,g,\id} \mid \NE h.\tuple{g,h,i} \in \psem{P}\},\\
&\psem{\aslong{P}}& = &\{\tuple{g,h,\id}\in P^*\mid\nexists h'.\tuple{h,h',\id}\in\psem{\Fix(P)}\}, \\
\end{array}\]
where $P^*=\bigcup_{j=0}^\infty P^j$ with $P^0 = \Skip$, $P^{j} = \tuple{\Fix(P); P^{j-1}}$ for $j>0$ 
and $\psem{\Fix(P)}=\{\tuple{g,h\circ i,\id}\mid \tuple{g,h,i}\in\psem{P}\}$. 
Instead of $\tuple{g,h,i}\in\psem{P}$, we write $g\dder_{P,i}h$ or short $g\dder_{P}h$.
\end{definition}

\longv{\red
\begin{remark}
The semantics of the sequential composition implies that a program with interface $X$ may only be iterated, if the output interface of the previous computation equals the (input) interface $X$.
The statement $\Fix$ is a generic way of making programs iterable that do not delete or unselect elements of their interface. $\Fix$ ensures that every possible computation ends with the output interface $X$ by finally deselecting all elements additionally selected during a run of the program. 
\end{remark}}


In the following, we consider the basic transformations \cite{Habel-Pennemann09a}.
The construction $\Shift$ ``shifts'' existential conditions over morphisms into a disjunction of existential application conditions. This can be done because the category of typed graphs has an $\Epi'- \M$-factorization. The construction $\Left$ ``shifts'' a right  application condition over a rule into a left application condition. Constraints can be integrated into left application conditions of a rule such that every transformation is condition-preserving. 


 \begin{lemma}[$\Shift,\Left,\cpres$]\label{lem:shift}\label{lem:left}\label{lem:pres} 
In an $\M$-adhesive category $\C$ with $\Epi'-\M$ pair factorization, there are constructions $\Shift$, $\Left$, and $\cpres$ such that the following holds.
For each condition $d$ over $P$ and every\ignore{\red type-refining} morphism $b\colon P\injto R,n\colon R\injto H$, $n\circ b\models d\iff n \models \Shift(b,d)$.
For each rule $p=\brule{L}{K}{R}$ and each condition $\ac$ over $R$, for each $G\dder_{p,g,h}H$, $g\models\Left(p,\ac)\iff h\models\ac$.
For each rule $\prule$ and each condition $d$, a condition $\ac=\cpres(\prule,d)$ can be constructed such that the  rule $\tuple{p,\ac}$ is $d$-\emph{preserving}, i.e., for all $g\dder_{\tuple{\prule,\ac}} h$, $g\models d$ implies $h\models d$.
\end{lemma}
A pair $(a',b')$ of morphisms is \emph{jointly surjective} if for each $x\in R'$ there is a preimage $y\in R$ with $a'(y)=x$ or $z\in C$ with $b'(z)=x$. For a rule $p=\brule{L}{K}{R}$, $p^{-1}=\brule{R}{K}{L}$ denotes the \emph{inverse} rule. For $L'\dder_p R'$ with intermediate graph $K'$, $\brule{L'}{K'}{R'}$ is the \emph{derived} rule.
\begin{construction}\label{const:Shift} \label{const:Left}\label{const:pres}The construction is as follows.
\[\begin{tabular}{l} 
\tikz[shape=rectangle,node distance=2em,shape=circle,outer sep=0pt,inner sep=1pt]{
\node(P){$P$};
\node(C)[strictly below of=P]{$C$};
\node(space)[node distance=1.5em,strictly below of=C]{};
\node(P')[strictly right of=P]{$R$};
\node(C')[strictly right of=C]{$R'$};
\draw[monomorphism] (P) -- node[overlay,left]{\small $a$} (C);
\draw[monomorphism,dashed] (P') -- node[overlay,right,inner sep=0pt]{\small $a'$} (C');
\draw[draw=none] (P) -- node[overlay]{(0)} (C');
\draw[monomorphism] (P) -- node[overlay,above]{\small $b$} (P');
\draw[monomorphism,dashed] (C) -- node[overlay,below]{\small $b'$} (C');
\node(c)[outer sep=0pt,inner sep=0pt,node distance=0em,strictly below of=C]{
\tikz[baseline,draw=black,fill=lightgray]{\filldraw (0,0) -- node[below,pos=0.6,overlay,outer sep=1ex]
{\small $c$} (0.12,-0.3) -- (-0.12,-0.3) -- (0,0);}};
\node(c')[outer sep=0pt,inner sep=0pt,node distance=0em,strictly below of=C']{
\tikz[baseline,draw=black,fill=lightgray]{\filldraw (0,0) -- node[below,pos=0.6,overlay,outer sep=1ex]
{\small } (0.12,-0.3) -- (-0.12,-0.3) -- (0,0);}};}\end{tabular}
\hspace{0.3cm}\hfill
\begin{tabular}{p{12cm}} 
$\Shift(b,\ctrue):=\ctrue$.\\
$\Shift(b,\PE(a,c)):=\bigvee_{(a',b')\in\F}\PE(a',\Shift(b',c))$ where \\
$\F=\{(a',b')\mid\mbox{$b'\circ a=a'\circ b$, $a',b'$ inj, $(a',b')$ jointly surjective}\}$\\
$\Shift(b,\neg d):=\neg\Shift(b,d)$, $\Shift(b,\wedge_{i\in I}d_i):=\wedge_{i\in I}\Shift(b,d_i)$.
\end{tabular}\]
\[\begin{tabular}{l}
\tikz[node distance=2em,shape=rectangle,outer sep=0pt,inner sep=2pt,label distance=0pt]{
\node(space){};
\node(L)[node distance=0em,strictly right of=space]{$R$};
\node(K)[strictly right of=L]{$K$};
\node(R)[strictly right of=K]{$L$};
\node(Kstar)[strictly below of=K]{$K'$};
\node(Lstar)[strictly below of=L]{$R'$};
\node(Rstar)[strictly below of=R]{$L'$};
\node(space)[node distance=1em,strictly below of=Rstar]{};
\draw[altmonomorphism] (K) -- node[overlay,above]{}(L);
\draw[monomorphism] (K) -- node[overlay,above]{} (R);
\draw[altmonomorphism] (Kstar) -- (Lstar);\draw[monomorphism] (Kstar) -- (Rstar);
\draw[morphism] (L) -- node[left]{$$} node[left]{\footnotesize $a$}(Lstar);
\draw[morphism] (K) -- (Kstar);
\draw[morphism] (R) -- node[left]{$$} node[right]{\footnotesize $a'$}(Rstar);
\draw[draw=none] (L) -- node[overlay]{\footnotesize (1)} (Kstar);
\draw[draw=none] (R) -- node[overlay]{\footnotesize (2)} (Kstar);
\node(space)[node distance=1.5em,strictly below of=Lstar]{};
\node(acL)[outer sep=0pt,inner sep=0pt,node distance=0em,strictly below of=Lstar]{
\tikz[baseline,draw=black,fill=lightgray]{
\filldraw (0,0) -- node[below,pos=0.7,overlay,outer sep=1ex]
{\footnotesize $\ac$} (-0.12,-0.3) -- (0.12,-0.3) -- (0,0);}};
\node(acR)[outer sep=0pt,inner sep=0pt,node distance=0em,strictly below of=Rstar]
{\tikz[baseline,draw=black,fill=lightgray]{\filldraw (0,0) -- node[below,pos=0.7,overlay,outer sep=1ex] 
{\footnotesize $$} (0.12,-0.3) -- (-0.12,-0.3) -- (0,0);}};}
\end{tabular}
\hspace{0.1cm}\hfill
\begin{tabular}{p{12cm}} 
$\Left(p,\ctrue):=\ctrue$.\\
$\Left(p,\PE(a,\ac)) := \PE(a',\Left(p',\ac))$ if $p^{-1}$ is applicable w.r.t. the morphism $a$, $p':=\brule{L'}{K'}{R'}$ is the \emph{derived}\, rule, and $\cfalse$, otherwise. 
$\Left(p,\neg\ac){:=}\neg\Left(p,\ac)$. $\Left(p,\wedge_{i\in I}\ac_i):=\wedge _{i\in I}\Left(p,\ac_i)$.
\end{tabular}\]

$\cpres(\prule,d):=\Shift(A \injto L,d)\impl\Left(\prule,\Shift(A \injto R,d)$.
\end{construction}

\begin{example}\label{exa:preservation}
Let $\prule = \tuple{\;\classlabindex{Pl}{}\classlabindex{Tk}{}\dder\containmentedgelabindex{Pl}{Tk}{\tok}{}{}\;}$ be a rule, $d = \NE (\scalebox{1}{\atmostonecontainerlab{Pl}{Tk}{Pl}{\tok}{\tok}})$ be a constraint, and 
$b_L$ and $b_R$ be the morphisms from the empty graph to the left- and right-hand side of the rule, respectively. Then we have the following.

\renewcommand{\arraystretch}{1.5}

\[\begin{array}{lcl}
\Shift(b_L, d) &=& \NE (\; \atmostonecontainerlabindex{Pl}{Tk}{Pl}{\tok}{\tok}{1}{2}\;) \wedge \ldots\\
\Shift(b_R, d) &=& \NE (\; \atmostonecontainerlabindex{Pl}{Tk}{Pl}{\tok}{\tok}{1}{2}\;) \wedge \ldots\\

\Left(\prule, \Shift(b_R, d)) &=& 
\NE  (\; \scalebox{0.8}{\begin{tikzpicture}[node distance=2em]
\node[node,label={[marking]below:\footnotesize 1}] (A) {$\mathrm{Pl}$};
\node[node,strictly right of=A,label={[marking]below: \footnotesize 2}] (B) {$\mathrm{Tk}$};
\node[node,strictly right of=B] (C) {$\mathrm{Pl}$};
\draw[contain] (C) to node[above] {\footnotesize tok} (B);
\end{tikzpicture}}\;) \wedge \ldots\\
\cpres(\prule, d) &=& 
\NE (\; \atmostonecontainerlabindex{Pl}{Tk}{Pl}{\tok}{\tok}{1}{2} \;) \wedge \ldots \dder
\NE  (\; 
\scalebox{0.8}{\begin{tikzpicture}[node distance=2em]
\node[node,label={[marking]below:\footnotesize 1}] (A) {$\mathrm{Pl}$};
\node[node,strictly right of=A,label={[marking]below: \footnotesize 2}] (B) {$\mathrm{Tk}$};
\node[node,strictly right of=B] (C) {$\mathrm{Pl}$};
\draw[contain] (C) to node[above] {\footnotesize tok} (B);
\end{tikzpicture}}\;) \wedge \ldots
\end{array}\]

If the rule $\prule$ is equipped with the application condition $\cpres(\prule, d)$, we obtain the rule $\prule' = \tuple{\prule,\cpres(\prule, d)}$, restricting the applicability of the rule to those matches satisfying the application condition and preserving the constraint $d$. 
The application condition is satisfied if the rule is applied to an occurrence of a graph with Tk-node hat does not have incoming containment edges from different Pl-nodes. (A node of type Tk is said to be \emph{Tk-node}.)

\end{example}


\section{Repair programs}\label{sec:repair}

In this section, we introduce repair programs and show some repair results for repair programs.

A repair program for a constraint is a graph program with the property that there exists a derivation and the application to any graph yields a graph satisfying the constraint. More generally, we consider repair programs for conditions.


\begin{definition}[repair programs]\label{def:repair} A (typed) program $P$ is a \emph{(typed) repair program} for a constraint $d$ if,
for all (typed) graphs $G$, $\exists\,G\dder_P H$ and $\forall\,G\dder_P
H$, $H\models d$.
A program $P$ with interface~$A$ is a \emph{repair program} for a
condition $\ac$ over $A$, if, for all injective morphisms $g\colon
A\injto G$, $\exists\,g\dder_{P,i} h$ and $\forall\,g\dder_{P,i} h$,
$h\circ i\models\ac$.
\end{definition}

\ignore{\begin{remark}The requirement of the existence of a
transformation is necessary: Let $d_1,d_2$ be constraints with $d_1=\neg
d_2$. Then $d=d_1\wedge d_2$ is not satisfiable. With the old
definition,  $P=\tuple{P_1;P_2^{d_1}}\equiv\Abort$ is a repair program
for $d$: There does not exist a transformation $G\dder H$. Thus,
$\forall G\dder H$, $H\models d$. \end{remark}} 

\begin{example}
For the constraint $d$ (see below), intuitively meaning, there do not exist two parallel edges of type $\token$ between a Pl-node and a Tk-node, the program $P_d$ is a repair program for $d$.
\[d= \NE (\;\paredgelab{Pl}{Tk}{\token}{\token} \;) \quad \quad P_d = \tuple{\; \scalebox{1}{\paredgelab{Pl}{Tk}{\token}{\token}} \dder \scalebox{1}{\referenceedgelab{Pl}{Tk}{\token}} \;}\downarrow\] 
It works as follows: whenever there are two parallel $\token$-edges, $P_d$ deletes one of the two $\token$-edges as long as possible.
\end{example}

We look for stable, maximally preserving, and terminating repair programs. 
A repair program is 
\emph{stable}, if it does nothing whenever the condition is already satisfied,
\emph{maximally preserving}, if, informally, items are preserved whenever possible (see \cite{Sandmann-Habel19a}),  
\emph{terminating} if the relation $\dder$ is terminating. 

{ 
We start with basic conditions of the form $\PE a$ (or $\NE a$) with morphism $a \colon A \injto C$. For basic conditions, we construct so-called \emph{repairing} sets from the morphism of the condition and repair programs based on the repairing set using the $\try$-statement and the as-long-as-possible iteration, respectively.
}

\ignore{
The construction of the repair programs is based on the construction for the \emph{basic} conditions requiring the existence (non-existence) of a morphism. 

\begin{definition}
A condition of the form $\PE a$ or $\NE a$ with real inclusion morphism $a \colon A \injto C$ is called \emph{basic}.
\end{definition}

Given a basic condition $\PE a ~(\NE a)$ with real inclusion morphism $a \colon A \injto C$, we construct \emph{repairing} rule sets $\R_a$ and $\S_a$, respectively. Whenever we have repairing sets, we obtain repair programs using the $\try$-statement and the as long as possible iteration, respectively.
}

\begin{lemma}[basic repair]\label{lem:basic}
For basic conditions, there are repair programs.
\end{lemma}

\begin{construction}\label{const:basic} 
For a real morphism $a\colon A\injto C$, the programs $P_{\PE a}$ and $P_{\NE a}$ are constructed as follows.
\renewcommand{\arraystretch}{1.5}
\[
\begin{array}{lll}
P_{\PE a} = \try \R_a & \mbox{with } \R_a=\{\tuple{b,B\dder C,\ac\wedge\ac_B,a}\mid A\injto^b B\subset C\}\\
P_{\NE a} = \S_a'{\downarrow} & \mbox{with } \S_a=\{\tuple{a,C\dder B,b}\mid A\injto^b B\subset C\mbox{ and (*)}\}
\end{array}
\]
where $\ac=\Shift(A\injto B,\NE a)$, $\ac_B{=}\bigwedge_{B\subset B'\subseteq C}\NE B'$, (*) $\pif \E_C\supset \E_{A} \pthen |\V_C|=|\V_B|,|\E_C|=|\E_B|+1 \pelse |\V_C|=|\V_B|+1$, and $'$~denotes the dangling edges operator.
\end{construction}

The rules in $\R_a$ are \ignore{increasing and }of the form $B\dder C$ where $A\subseteq B\subset C$. They possess an application condition $\ac$ requiring the condition $\NE a$, shifted from $A$ to~$B$, and the application condition $\ac_B$ requiring that no larger subgraph $B'$ of~$C$ occurs. The rules in $\S_a$ are \ignore{decreasing and }of the form $C\dder B$ where $A\subseteq B\subset C$ such that, if the number of edges in $C$ is larger than the one in $A$, they delete one edge and no node, and delete a node, otherwise. By $B\subset C$, both rule sets
do not contain identical rules.

\begin{example}\label{ex:Ra}
Consider the condition $d= \PE b$ with $b\colon\classlab{Pl}\injto \containmentedgelab{Pl}{Tk}{\tok}$, intuitively meaning that, whenever there is a place there exists a token and a connecting containment edge.
Application of the Construction~\ref{const:basic} yields a rule set $\R_b$ with two rules.
\renewcommand{\arraystretch}{1.3}
\[\R_b=
\left\{\begin{array}{lcl}
\prule_1=\tuple{\;x_1, \classlab{Pl}&\dder&\containmentedgelab{Pl}{Tk}{\token},\NE \classlab{Pl} ~\classlab{Tk},y_1\;}\\
\prule_2=\tuple{\;x_2, \classlab{Pl}\classlab{Tk}&\dder&\containmentedgelab{Pl}{Tk}{\token},
\NE\containmentedgelab{Pl}{Tk}{\token} \wedge \NE \classlab{Tk}\containmentedgelab{Pl}{Tk}{\token}, y_2\;}\\
\end{array}\right.\]
where 
$x_1 \colon \classlab{Pl} \injto \classlab{Pl}$, $y_1 \colon \classlab{Pl} \injto \containmentedgelab{Pl}{Tk}{\token}$, $x_2 \colon \classlab{Pl} \injto \classlab{Pl} \classlab{Tk}$, and $y_2 \colon \classlab{Pl} \injto \containmentedgelab{Pl}{Tk}{\token}$.
The rule $\prule_1$ requires a node of type Pl and attaches a node of type Tk and a connecting containment edge, provided that there do not exist a Pl-node and a Tk-node. 
The second rule $\prule_2$ requires an occurrence of a Pl- and a Tk-node and inserts a connecting containment edge, provided there is no containment edge from the occurrence of the Pl-node to the image of Tk-node, and there is no containment edge to another Tk-node.
By Lemma~\ref{lem:basic} resp. Theorem \ref{thm:repair}(2), $P_d=\try\R_b$ is a repair program for $d=\PE b$.
\end{example}

Conditions with alternating quantifiers ending with $\ctrue$ or of the form $\PE(a,\NE b)$ or $\NE b$ are \emph{proper}. 
A proper condition of the form $\PA(a,c)$ and $\PE(a,c)$ that ends with $\ctrue$ is \emph{universal} and \emph{existential}, respectively. A condition of the form $\PE a$ ($\NE a$) is \emph{positive} (\emph{negative}).

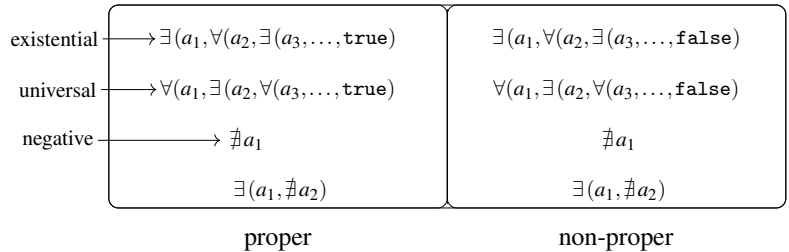
\[\scalebox{0.9}{
\tikz[node distance=1.2em,label distance=1pt,outer sep=1pt]{%
  \node(l1) at (2.5,0.5){};
  \node(l2) at (2.5,-2.5){};
  \draw[-] (l1) to (l2); 
  \node(uld) at (-2.5,-2.5){};
  \node(org) at (7.5,0.5) {};
 \draw[rounded corners] (uld) node [above]{} rectangle (org); 
 \draw[rounded corners] (uld) rectangle (l1);
 \draw[rounded corners] (org) rectangle (l2);

 \node(0) at (0,0) {\footnotesize $\PE(a_1, \PA(a_2, \PE(a_3, \ldots, \ctrue)$};
 \node(1a) at (5,0) {\footnotesize $\PE(a_1, {\PA}(a_2, \PE(a_3, \ldots, \cfalse)$};
 \node(1b) at (5,-0.75) {\footnotesize $\PA(a_1, {\PE}(a_2, \PA(a_3, \ldots, \cfalse)$};
 
 \node(PA) at (0,-0.75) {\footnotesize $\PA(a_1, \PE (a_2, \PA(a_3, \ldots, \ctrue)$};
 \node(NE) at (-0.5,-1.5) {\footnotesize \ignore{+} $ \NE a_1$};   
 \node(2) at (0,-2.25) {\footnotesize \ignore{+} $\PE (a_1, \NE a_2)$};   
 \node(3) at (5,-2.25) {\footnotesize \ignore{-} $\PE (a_1, \NE a_2)$};  
 \node(4) at (5,-1.5) {\footnotesize \ignore{-} $\NE a_1$};     
 
\node(uni) at (-3.25,-0.75) {\footnotesize universal};
\node(neg) at (-3.25,-1.5) {\footnotesize negative};
\node(exi) at (-3.3,-0) {\footnotesize existential};
 
  \node (proper) at (0,-3) {proper};
  \node (r) at (5,-3) {non-proper};
  \draw[->] (uni) to (PA);
  \draw[->] (neg) to (NE);  
  \draw[->] (exi) to (0);  
  }} \]

For proper conditions, a repair program can be constructed.

\begin{theorem}[Repair I]\label{thm:repair} There is a repair program for proper conditions.
\end{theorem}
\begin{construction}\label{const:proper} For proper conditions $d$, the repair program $P_d$ is constructed inductively as follows.
\begin{enumerate}
\item[(1)] For $d = \ctrue$, $P_d = \Skip$.
\item[(2)] For $d = \PE a$, $P_d = \try\R_a$.
\item[(3)] For $d = \NE a$, $P_d = \Rdown{\S_a'}$.
\item[(4)] For $d = \PE(a,c)$, $P_d = P_{\PE a};\tuple{\select(a);P_c;\unselect(a)}$.
\item[(5)] For $d = \PA(a,c)$, $P_d = \Rdown{\tuple{\select(a,\neg c);P_c;\unselect(a)}}$.
\end{enumerate}
where $a\colon A\injto C$ is real, $\R_a$ and $\S_a^\prime$ are the sets according to Construction~\ref{const:basic}, and $P_c$ is a repair program for $c$ with interfaces $C$.
$\select(a)=\tuple{a,\id_C}$ is the rule with left interface $a$ and identical plain rule $\id_C=\brule{C}{C}{C}$. 
Given an occurrence of $A$, it is used for a marking of an occurrence of~$C$, extending the occurrence of $A$.
Similar, $\select(a,\ac)=\tuple{a,\id_C,\ac}$ is used for marking an occurrence of~$C$ satisfying the condition $\ac$. $\unselect(a)=\tuple{\id_C,a}$ is the identical plain rule with right interface $a$, used for unmarking the  occurrence of $C$.
\end{construction}

\begin{example}\label{exa:repairprogram}
Given the constraint $d= \PA(\scalebox{1}{\classlab{Pl}}, \PE \scalebox{1}{\containmentedgelab{Pl}{Tk}{\token}})$, meaning that, for each place, there exists a token, a repair program for~$d$ can be constructed according to Theorem~\ref{thm:repair}.  
The constraint $d$ is of the form $\PA(a,c)$ with morphism $a\colon \emptyset\injto \classlab{Pl}$ and condition $c=\PE\classlab{Pl}\injto\scalebox{1}{\containmentedgelab{Pl}{Tk}{\token}}$. By Theorem \ref{thm:repair}(5), $P_d = \tuple{\select(a, \neg c);P_{c}; \unselect(a)}\downarrow$. 
The condition $c$ is of the form $\PE b$. By Repair Theorem~\ref{thm:repair}(2), the repair program for $c$ is $P_c = \try \R_b$, where $\R_b$ is the rule set from Example~\ref{ex:Ra}.
%
The program $P_d$ marks an occurrence of a Pl-node without connecting containment edge to a Tk-node. 
The program $P_c$ tries to add a Tk-node and the containment edge, provided there do not exist a Pl-node and a Tk-node, and to add a containment edge between a Pl- and Tk-node, provided that there does not exist such an edge to another Tk-node. 
Finally the marked part is unmarked. This is done as long as possible. 
Whenever no further application is possible, the constraint $d$ is satisfied.
\end{example}

The constructed programs are increasing (decreasing): 
A program $P$ is \emph{decreasing (increasing)} if all rules in $P$ are decreasing (increasing).
A rule $\prule=\tuple{L\injlto K \injto R,\ac}$ is \emph{decreasing} if $L\supset K\cong R$ and increasing if $L\cong K\subset R$.

\begin{fact}\label{fac:decrease}
For negative conditions, the repair program is decreasing. For positive, existential, and universal conditions, the repair program is increasing.
\end{fact}

\ignore{ 
\begin{proof}[of Theorem \ref{thm:repair}]
The proof in \cite[Theorem 1]{Habel-Sandmann18a} can be extended to typed graphs: 
For every type-refining morphism $a\colon A\injto C$, and every proper subgraph $B$ of $C$, define $type_B = type_C \circ \inc_B$, where for $B\subseteq C$, $\inc_B$ (short $i_B$) denotes the inclusion of $B$ in $C$. Then the inclusion morphism $\inc_A\colon A\injto B$ is type-refining, $\inc_B\colon B\injto C$ is type-preserving, and the rules $B\dder C\in\R_a$ and $C\dder B\in\S_a$ consist of type-preserving morphisms. In this way the graphs in the rules in $\R_a$ and $\S_a$ become typed. The typing of the conditions is a direct consequence.
\ignore{
\[\tikz[node distance=3em,shape=rectangle,outer sep=1pt,inner sep=2pt]{
\node(A){$A$};
\node(B)[strictly above right of=A]{$B$};
\node(C)[strictly below right of=B]{$C$};
\node(TG)[strictly above of=B]{$TG$};
\draw[monomorphism] (A) -- node[overlay,above](a){$a$} (C);
\draw[monomorphism] (A) -- node[overlay,above]{$i_A$} (B);
\draw[monomorphism] (B) -- node[overlay,above]{\footnotesize{$i_B$}} (C);
\draw[draw=white] (a) -- node[overlay]{=} (B);
\draw[monomorphism] (A) -- node[overlay,left](a){\footnotesize{$type_A$}} (TG);
\draw[monomorphism] (B) -- node[below left](b){\footnotesize{\ignore{$\leq$}}} node[below right](b){\ignore{=}} (TG);
\draw[monomorphism] (C) -- node[overlay,right]{\footnotesize{$type_C$}} (TG);}
\vspace{-0.4cm}\]}
\end{proof}
}

{ 
\begin{proof}[of Theorem \ref{thm:repair}]
In  \cite[Theorem 1]{Habel-Sandmann18a} the statement is proven for graphs. 
The statement also holds for typed graphs: 
For every morphism $a\colon A\injto C$, and every proper subgraph $B$ of $C$, define $type_B = type_C \circ \inc_B$, where for $B\subseteq C$, $\inc_B$ (short $i_B$) denotes the inclusion of $B$ in $C$. 
Then $type_A = type_B \circ \inc_A$ and the rules $B\dder C\in\R_a$ and $C\dder B\in\S_a$ consist of typed graph morphisms. In this way the graphs in the rules in $\R_a$ and $\S_a$ become typed. The typing of the conditions is a direct consequence.
\ignore{
\[\tikz[node distance=3em,shape=rectangle,outer sep=1pt,inner sep=2pt]{
\node(A){$A$};
\node(B)[strictly above right of=A]{$B$};
\node(C)[strictly below right of=B]{$C$};
\node(TG)[strictly above of=B]{$TG$};
\draw[monomorphism] (A) -- node[overlay,above](a){$a$} (C);
\draw[monomorphism] (A) -- node[overlay,above]{$i_A$} (B);
\draw[monomorphism] (B) -- node[overlay,above]{\footnotesize{$i_B$}} (C);
\draw[draw=white] (a) -- node[overlay]{=} (B);
\draw[monomorphism] (A) -- node[overlay,left](a){\footnotesize{$type_A$}} (TG);
\draw[monomorphism] (B) -- node[below left](b){\footnotesize{$=$}} node[below right](b){=} (TG);
\draw[monomorphism] (C) -- node[overlay,right]{\footnotesize{$type_C$}} (TG);}
\vspace{-0.4cm}\]}
\end{proof}
}

\begin{remark}\label{rem:generalized-proper}Theorem
\ref{thm:repair} could  be formulated for a larger class of conditions.
In Construction \ref{const:proper}, it is not necessary that the condition $c$ in a condition $\PE(a,c)$ (or $\PA(a,c)$) is proper. The construction of a repair program can be done provided that there exists a repair program for $c$. Properness only guarantees the existence of a repair program.\end{remark}

In the following, we consider conjunctions of conditions. 
We try to construct a repair program for a conjunction from the repair programs of the conditions in the conjunction.

Given a conjunction $d$ of conditions, we proceed as follows.
\begin{enumerate}
\item[(1)] Try to find a ``preserving sequentialization'' $d_1, \ldots, d_n$ of $d$.
\item[(2)] Construct repair programs $P_1,\ldots,P_n$ for $d_1, \ldots, d_n$. 
\item[(3)] Compose the repair programs to a repair program $P=\tuple{P_1;\ldots;P_n}$ for $d$.
\end{enumerate}

\ignore{\red The construction of repair programs relies on finding a preserving sequentialization of a given condition.  
The process is to consider all sequentializations and test it on preservation. In general, for a constraint and a program, it is undecidable if the program is constraint-preserving \cite[Theorem 6]{Sager19a}. Consequently, the problem of finding a preserving sequentialization is undecidable. 

In general, the programs cannot be applied in arbitrary order (see Example~\ref{exa:preserving1}) and sometimes cannot be ordered (see Example~\ref{exa:no-sequence}).
The sequential composition $P=\tuple{P_1;\ldots;P_n}$ of the programs may be no repair program for $d$:
Since $P_1$ is a repair program for $d_1$, after the application of $P_1$, the condition $d_1$ is satisfied.
Since $P_2$ is a repair program for $d_2$, after the application of $P_2$, the condition $d_2$ is satisfied.
But the condition $d_1$ may be not preserved. Thus, we have to look for a preserving sequentialization of the conditions resp. repair programs. 
}

1. For a conjunction of negative (positive) conditions, this is very simple. We take any sequentialization of the negative (positive) conditions and consider the sequential composition of the corresponding repair programs. 
This works because, for negative (positive) conditions, the repair programs are decreasing (increasing), and every sequence of decreasing (increasing) repair programs ``preserves'' the preceding negative (positive) conditions.

2. For conjunctions of universal conditions, this is not so easy: In general, not every sequentialization is preserving. Consider, e.g.,  the constraints $d_1=\PA(\onenode{},\PE\onenodeloop{}{})$ and $d_2=\PA(\onenode{},\PE\twonodesedge{}{}{})$ with the repair programs $P_1$ and $P_2 $ constructed according to Construction~\ref{const:proper}. For the sequentialization $d_1,d_2$, the program $P_2$  does not preserve the constraint $d_1$: For a node with loop\ satisfying $d_1$ the program $P_2$ adds a new node and a connecting edge. The new node does not have a loop, i.e., the resulting graph does not satisfy $d_1$. For the sequentialization $d_2,d_1$,  the program $P_1$  preserves the constraint $d_2$.

3. Moreover, sometimes there is no preserving sequentialization.
Consider, e.g., $d_1=\PA(\onenode{},\PE\twonodesedge{}{}{})$ and $d_2=\PA(\twonodesedge{}{}{},\PE\twonodesedge{}{}{}\onenode{})$ with the repair programs $P_1$ and $P_2 $ constructed according to Construction \ref{const:proper}.
The condition $d_1\wedge d_2$ is satisfiable: the graph $\threenodesthreecycle{}{}{}{}{}$ satisfies $d_1 \wedge d_2$.
Then $P_2$ does not preserve $d_1$ and $P_1$ does not preserve $d_2$:
Application of $P_2$ to $\twonodesedgeredge{}{}{}{}\models d_1$ yields to $\twonodesedgeredge{}{}{}{}\onenode{}\not\models d_1$ and application of $P_1$ to $\onenode{}\models d_2$ yields to $\twonodesedgeredge{}{}{}{}\not\models d_2$. 
\ignore{
\[\begin{array}{lll}
\onenode{}\DSLongdder_{P_1}\twonodesedgeredge{}{}{}{}&\DSLongdder_{P_2}\twonodesedgeredge{}{}{}{}\onenode{}&\not\models d_1\\
\onenode{}\DSLongdder_{P_2}\onenode{}&\DSLongdder_{P_1}\twonodesedgeredge{}{}{}{}&\not\models d_2\\
\end{array}\]}
\vspace{0.3cm}

For this proceeding, preservation of conditions is essential: Whenever a condition is satisfied, it shall be preserved in the following.

\begin{definition}[preservation]\label{def:preservation}
A program $P$ is \emph{$d$-preserving} if every rule in $P$ is $d$-preserving.
\end{definition}

\begin{lemma}[preservation]\label{lem:preserve} 
For every program $P$ and every $d$, there is a $d$-preserving program~$P^d$.
\end{lemma}

\begin{construction}$P^d$: replace all rules $\prule$ in $P$ by $\tuple{\prule,\cpres(\prule,d)}$ (Construction~\ref{lem:pres}).
\end{construction}

\begin{proof}
By Construction \ref{lem:pres}, all rules in $P^d$ are $d$-preserving, thus, the program is $d$-preserving.
\end{proof}

In the following, we consider a sequence of conditions together with their repair programs.
\begin{convention}
Let $ds = d_1, \ldots, d_n$, $Ps = P_1, \ldots P_n$, and $P_i$ be a repair program for $d_i$, respectively.
\end{convention}

A~sequence of programs is preserving if for each natural number~$k$, the respective repair program $P_k$ preserves all preceding conditions, i.e. $P_1$ is a repair program for $d_1$, $P_2$ is a repair program for $d_2$ and $d_1$-preserving, $P_3$ is a repair program for $d_3$ and $d_1 \wedge d_2$-preserving, and so on\longv{(see Figure~\ref{fig:preservation})}.  

\begin{definition}[preservation]
The sequence $Ps$ is \emph{$ds$-preserving} (and the sequence $ds$ is \emph{preserving}) if, for $k=2,\ldots,n$, $P_k$ is $\wedge_{i=1}^{k{}-1} d_i$-preserving. 
\end{definition}

We show that sequences of repair programs for sequences of conditions can be sequentially composed to a repair program for the conjunction, provided that the sequences of conditions is preserving.

\begin{lemma}[preserving repair]\label{lem:presrepair}
{If} $d_1,\ldots, d_n$-preserving, {then} $\tuple{P_1;\ldots;P_n}$ is a repair program for $\wedge_{i=1}^n d_i$. 
\end{lemma}

\begin{proof}
By induction on the number $n$ of conditions with the repair programs. 
For $n{=}1$, by Theorem~\ref{thm:repair}, $P_1$ is a repair program for $d_1$.
{Inductive hypothesis:} If $P_2,\ldots,P_n$ is $d_1,\ldots,d_n$-preserving, then $P=\tuple{P_1;\ldots;P_n}$ is a repair program for the conjunction $\wedge_{i=1}^n d_i$. 
{Inductive step:} For $n=n+1$, let $P_2,\ldots,P_{n+1}$ be $d_1,\ldots,d_{n{+}1}$-preserving. 
Then $P_2,\ldots,P_n$ is $d_1,\ldots d_n$-preserving and, by induction hypothesis, the program $P=\tuple{P_1;\ldots;P_n}$ is a repair program for the conjunction $d=\wedge_{i=1}^n d_i$. Moreover, $P_{n+1}$ is a $d$-preserving repair program for $d_{n+1}$. 
Consequently, for every transformation $g\dder_{P} g_n \dder_{P_{n+1}}h$, $g_n\models d$ and $h \models \wedge_{i=1}^{n+1} d_i$.
Thus, $\tuple{P_1;\ldots;P_{n+1}}$ is a repair program for $\wedge_{i=1}^{n+1} d_i$. 
\end{proof}

\ignore{\red
\begin{definition}
A sequence is $ds = d_1, \ldots, d_n$ is \emph{homogeneous}, if all $d_i$ are either positive, or negative, or preserving.   
\end{definition}
}

A sequence $ds$ of conditions is \emph{negative (or positive, or universal)} if all conditions in it have the property.

For a sequence of negative (or positive) conditions, the sequence $Ps$ of repair programs is $ds$-preserving.

\begin{fact}\label{fac:neg}If $ds$ is negative (or positive), then $Ps$ is $ds$-preserving.
\end{fact}
\ignore{
\begin{proof}
For a sequence of negative (or positive) conditions, the underlying repair program is decreasing (or increasing), thus, the sequence is preserving. 
\end{proof}
}

In the following, we consider a conjunction of negative and universal conditions. Let $e_1$ be the conjunction of negative and $e_2$ the conjunction of universal conditions. By Fact \ref{fac:neg}, every sequence $ds_1$ of negative conditions is preserving and, by Lemma \ref{lem:presrepair}, the sequential composition $Q_1=\tuple{P_{1};\ldots;P_k}$ of the repair programs forms a repair program for $e_1$. 
In general, \ignore{By {\red Example \ref{exa:preserving1}},} not every sequentialization $ds_2$ of universal conditions is preserving. We have to require preservation.
In the case of preservation, the sequential composition $Q_2=\tuple{P_{k+1};\ldots;P_n}$ of the repair programs forms a repair program for $e_2$. 
But the repair program $Q_2$ may be not $e_1$-preserving. By Lemma \ref{lem:repairing} below, $Q_2$ can be modified to an $e_1$-preserving repair program $Q_2^{\prime e_1}$ for $e_2$. The idea is to delete all occurrences of the morphism $a$ of the universal condition $\PA (a,c)$, which violate the condition $c$. Given a universal condition, we mark an occurrence of the morphism violating the condition, then the occurrence of the morphism at that position is deleted. To mark the morphism at that position, we modify the left interface to the identity of the codomain of the morphism. For a conjunction of universal conditions, an $e_1$-preserving repair program can be constructed from the $e_1$-preserving program by destroying all non-repaired occurrences of all the universal conditions.
By Lemma \ref{lem:presrepair}, the program $\tuple{Q_1,\Q_2^{\prime e_1}}$ becomes a repair program for $e_1\wedge e_2$.
 
A conjunction is \emph{negative (universal)} if all conditions in the conjunction are negative (universal).
\begin{lemma}[preserving repair program]\label{lem:repairing}
If $Q_2$ is a repair program for a preserving universal conjunction~$e_2$ and $e_1$ is a negative conjunction,  
then there is an $e_1$-preserving repair program $Q_2^{\prime e_1}$ for $e_2$.
\end{lemma}

\begin{construction}\label{const:repairing} For a repair program $P$ for $\PA(a,c)$, let 
$P^{\prime e_1}=\tuple{P^{e_1}; P_{\NE a}^\id}$
where $P_{\NE a}^\id = \tuple{\select(a,\neg c);\S_a^\id}\downarrow$ and $\S_a^\id$ is obtained from $\S_a$ by replacing the left interface morphism $a\colon A\injto C$ by the identity $\id\colon C\injto C$.
%
For a  preserving conjunction $e_2$ of universal conditions with sequentialization $d_1,\ldots,d_n$ and repair program $Q_2=\tuple{P_1;\ldots;P_n}$, let $Q_2^{\prime e_1}=\tuple{P_1^{\prime e_1};\ldots;P_n^{\prime e_1}}$.
\end{construction}

\begin{proof}1. For a universal condition $\PA(a,c)$ with $a\colon A\injto C$, the programs $P$ and $P^{e_1}$ are increasing. 
By the $e_1$-preserving application condition, whenever the increasing program $P^{e_1}$ is not a repair program, the condition $c$ is not satisfied, and the decreasing program $P_{\NE a}^\id$ becomes applicable and destroys all occurrences that do not satisfy the condition $c$.
Since $e_1$ is a conjunction of negative
conditions, $P_{\NE a}^\id$ is $d_1$-preserving. Consequently, $P^{\prime e_1}$ is $e_1$-preserving. Then $P^{\prime e_1}$ is a repair program for $d$: For every occurrence of~$a$, the occurrence is either (1) repaired by $P^{e_1}$ or (2) destroyed by $P_{\NE a}^\id$.

2.~By assumption, $d_1,\ldots,d_n$ is preserving. Moreover, for $i=1,\ldots,n$, $P_i$ is repair program for $d_i$ and, by Lemma~\ref{lem:repairing}.1, $P_i^{\prime e_1}$ is an $e_1$-preserving repair program for~$d_i$.
Since $d_1,\ldots,d_n$ is preserving and by Lemma~\ref{lem:presrepair}, $Q_2^{\prime e_1} = \tuple{P_1^{\prime e_1};\ldots;P_n^{\prime e_1}}$ is a repair program for $\wedge_{i=1}^{n}d_i=e_2$. Since all programs in the sequential composition are $e_1$-preserving, the program $Q_2^{\prime e_1}$ is $e_1$-preserving.
Consequently, every transformation $g\dder_Q m$ is of the form $g\dder_{Q_1}h\dder_{Q_2^{\prime e_2}} m$. Since $Q_1$ is a repair program for $e_1$, $h\models e_1$. Since $Q_2^{\prime e_2}$ is the $e_1$-preserving repair program for $e_2$, $m\models e_1\wedge e_2$. Thus, $Q$ is a repair program for $e_1\wedge e_2$. 
\end{proof}


\begin{fact}[composition]\label{fac:composition}For arbitrary conditions $e_1,e_2$, the following holds.
If $Q_1$ is a repair program for $e_1$ and $Q_2^{\prime e_1}$ be an $e_1$-preserving repair program for $e_2$, then $\tuple{Q_1,Q_2^{\prime e_1}}$ a repair program for $e_1\wedge e_2$.
\end{fact}


\longv{
\begin{remark}
In general, the program ${P'_2}^{d_1}$ is not increasing.
\end{remark}
}

A sequence $ds = d_1, \ldots, d_n$ is negative (or positive, or existential, or universal) if all $d_i$ are negative (or positive, or existential, or universal).
In the following, $ds_1 = d_1, \ldots, d_k$, $ds_2 = d_{k+1}, \ldots, d_n$ with conjunction $e_1 = \wedge_{i=1}^{k} d_i$ and $e_2 = \wedge_{i=k+1}^{n} d_i$.

The following theorem says under which conditions a repair program for a conjunction of conditions can be constructed from the repair programs of its components. 

\begin{theorem}[Repair II]\label{thm:repair2}There is a repair program $P$ for a conjunction $d=\wedge_{i=1}^{n} d_i$ of conditions provided that $d$ is satisfiable, there are repair programs $P_1, \ldots, P_n$ for $d_1, \ldots, d_n$, respectively, and there is a sequentialization $ds = d_1, \ldots, d_n$, and
\begin{enumerate}
\item $ds$ is negative, or positive, or preserving,
\item $ds_1$ is positive, and $ds_2$ is existential (or universal) \& preserving. 
\item $ds_1$ is negative, and $ds_2$ is universal \& preserving.
\end{enumerate}
\end{theorem}

\begin{construction}
\begin{enumerate}
\item For negative (or positive, or preserving) $ds$, let $P=\tuple{P_1;\ldots;P_n}$.
\item For positive  $ds_1$, universal (or existential) \& preserving $ds_{2}$, let $P=\tuple{Q_1;Q_2}$.
\item For negative $ds_1$, universal \& preserving $ds_{2}$, let 
$P=\tuple{Q_1;Q_2^{\prime e_1}}$.
\end{enumerate}
where $P_1,\ldots,P_n$ are repair programs for $d_1,\ldots,d_n$, respectively, $Q_1{=}\tuple{P_1;\ldots;P_k}$, $Q_2=\tuple{P_{k{+}1};\ldots;P_n}$, and $Q_{2}^{\prime e_1}=\tuple{P_{k{+}1}^{\prime e_1};\ldots;P_n^{\prime e_1}}$ where $e_1{=}\wedge_{i=1}^k~d_i$ and $e_2=\wedge_{i=k+1}^n d_i$.
\end{construction}

\begin{example}\label{exa:conjunctive}
Consider the constraints 
$d_1 = \NE(\;\scalebox{1}{\atmostonecontainerlab{Pl}{Tk}{Pl}{\token}{\token}}\;)$ and 
$d_2=\PA (\scalebox{1}{\classlab{Pl}}, \PE (\scalebox{1}{\containmentedgelab{Pl}{Tk}{\token}}))$
(see Example \ref{exa:repairprogram}). 
By Theorem~\ref{thm:repair}, there are repair programs 
\[\begin{array}{lcl}
P_1 &= &\tuple{\;\scalebox{1}{\atmostonecontainerlab{Pl}{Tk}{Pl}{\token}{\token}} \dder \scalebox{0.8}{
\begin{tikzpicture}[node distance=2em]
\node[node] (A) {$\mathrm{Tk}$};
\node[node,strictly right of=A] (B) {$\mathrm{Pl}$};
\node[node,strictly left of=A] (C) {$\mathrm{Pl}$};
\draw[contain] (B) to node[above] {\footnotesize $\mathrm{\token}$} (A);
\end{tikzpicture}}\;}\downarrow\\
P_2&=&\tuple{\select(a, \neg c);\try \R_b; \unselect(a}\downarrow\\
\end{array}\] 
where $a\colon\emptyset\injto\classlab{Pl}$, $c=\PE{\classlab{Pl}}\injto \containmentedgelab{Pl}{Tk}{\token}$, and $\R_b$ as in Example~\ref{ex:Ra}. 
\ignore{\[\R_b=\left\{\begin{array}{lcl}
\prule_1=\tuple{\;x_1, \classlab{Pl}&\dder&\containmentedgelab{Pl}{Tk}{\token},\NE \classlab{Pl} ~\classlab{Tk},y_1\;}\\
\prule_2=\tuple{\;x_2, \classlab{Pl}\classlab{Tk}&\dder&\containmentedgelab{Pl}{Tk}{\token},
\NE\containmentedgelab{Pl}{Tk}{\token} \wedge{\red \NE}\classlab{Tk}\containmentedgelab{Pl}{Tk}{\token}, y_2\;}\\
\end{array}\right.\]}

By Theorem \ref{thm:repair2}, there is a repair program $P=\tuple{P_1;P_2^{\prime d_1}}$ for $d=d_1\wedge d_2$. 
By Lemma \ref{lem:preserve}, the $d_1$-preserving version $P_2^{d_1}$ of $P_2$ is obtained from the repair program $P_2$ by equipping each rule $\prule$ in$\R_b$ with application condition $\cpres(\prule,d_1)$. For $\prule_2\in\R_b$, $\prule_2' = \tuple{\prule_2,\cpres(\prule_2, d_1)}
= \NE (\; \atmostonecontainerlabindex{Pl}{Tk}{Pl}{\tok}{\tok}{}{}\wedge \ldots\; \dder
\NE(\scalebox{0.8}{\begin{tikzpicture}[node distance=2em]
\node[node,label={[marking]below:}] (A) {$\mathrm{Pl}$};
\node[node,strictly right of=A,label={[marking]below:}] (B) {$\mathrm{Tk}$};
\node[node,strictly right of=B] (C) {$\mathrm{Pl}$};
\draw[contain] (C) to node[above] {\footnotesize tok} (B);
\end{tikzpicture}})\wedge \ldots$.
The program $P_2^{d_1}$ is not a repair program for~$d_2$: $P_2^{d_1}$ is not applicable to the graph $G$:
$\scalebox{0.8}{\begin{tikzpicture}[node distance=2em]
\node[node,label={[marking]below:}] (A) {$\mathrm{Pl}$};
\node[node,strictly right of=A,label={[marking]below:}] (B) {$\mathrm{Tk}$};
\node[node,strictly right of=B] (C) {$\mathrm{Pl}$};
\draw[contain] (C) to node[above] {\footnotesize tok} (B);
\end{tikzpicture}}$ and 
$G\not\models d_2$. 

By Lemma \ref{lem:repairing}, there is a $d_1$-preserving repair program $P_2^{\prime d_1} = \tuple{P_2^{d_1}; P_{\NE a}^\id}$ for $d_2$ where 
$P_{\NE a}^\id$ is a slightly modified version of the repair program $P_{\NE a}$ for the condition $\NE a$. 
In more detail, the program looks as follows: $P_{\NE a}^\id = \tuple{\select(\classlab{Pl}, \NE \containmentedgelab{Pl}{Tk}{\token}); \tuple{x,\classlab{Pl} \dder \emptyset}'}{\downarrow}$ where $x$ is the identity $x \colon \classlab{Pl} \injto \classlab{Pl}$. The program marks an occurrence of the Pl-node without incoming containment edge from a Tk-node and deletes (in SPO-style) the occurrence of the Pl-node; this is done as long as possible.
In this way, Pl-nodes not satisfying the condition $c$, are deleted. We obtain a repair program for $d_1\wedge d_2$.
\end{example}



\begin{proof} 
1. Let $d_1,\ldots,d_n$ be negative (positive). Then the repair programs $P_1,\ldots,P_n$ are decreasing (increasing) and $d_1,\ldots,d_n$ is preserving. Then, by Lemma~\ref{lem:presrepair}, $\tuple{P_1;\ldots;P_n}$ is a  repair program for $\wedge_{i=1}^n d_i$. 

2. Let $d_1,\ldots,d_k$ be positive and $d_{k+1},\ldots,d_n$ universal (or existential) and preserving. 
By Theorem \ref{thm:repair2}.1 there are repair programs 
$Q_1=\tuple{P_1;\ldots;P_k}$ and $Q_2=\tuple{P_{k{+1}};\ldots;P_{n}}$ for $e_1=\wedge_{i=1}^k d_i$ and $e_2=\wedge_{i=k+1}^n d_i$, respectively. Since $Q_2$ is increasing, it is $e_1$-preserving. By Lemma~\ref{lem:presrepair}, $\tuple{Q_1;Q_{2}}$ is a repair program for $e_1\wedge e_2=\bigwedge_{i=1}^nd_i$. 

3. Let $d_1,\ldots,d_k$ be negative and $d_{k+1},\ldots,d_n$ universal and preserving. 
By Theorem \ref{thm:repair2}.1 there are repair programs 
$Q_1=\tuple{P_1;\ldots;P_k}$ and $Q_2=\tuple{P_{k{+1}};\ldots;P_{n}}$ for $e_1=\wedge_{i=1}^k d_i$ and $e_2=\wedge_{i=k+1}^n d_i$, respectively. 
By Lemma~\ref{lem:repairing}, $Q_{2}^{\prime e_1}$ is the $e_1$-preserving repair program for $e_2$.
By Fact~\ref{fac:composition}, $\tuple{Q_1;Q_{2}^{\prime e_1}}$ is a repair program for $e_1\wedge e_2=\bigwedge_{i=1}^nd_i$. 
An illustration of this part of the proof is given in Figure~\ref{fig:illust}.\end{proof}

\ignore{\red My problem for the third case is: Since $P_{\NE a}(k+2)$ is decreasing, and $d_{k+1}$ is universal, we cannot proof, that the repaired occurrence of $d_{k+1}$ is not again destroyed be the deleting program $P_{\NE a}(k+2)$.
On the other side, we have no counter example for that. 
I think, that should be the case if: for $d = d_1 \wedge d_2 \wedge d_3 = \NE a_1 \wedge \PA(a_2, c_2) \wedge \PA(a_3, c_3)$, with:
\begin{itemize}
\item[(1)] $d_2, d_3$ are independent,
\item[(2)] $P_3$ introduced violation of $d_1$
\item[(3)] $P_{\NE a_3}$ deletes some already repaired parts of $c_2$.  
\end{itemize}
Solution: We require, that the $\prime$ program preserves the universal condition. Better, it has already been required. For a sequence of program, not repair programs, the $P_i$ has to preserve $d_j (j<i)$
}
\begin{figure}[htp]
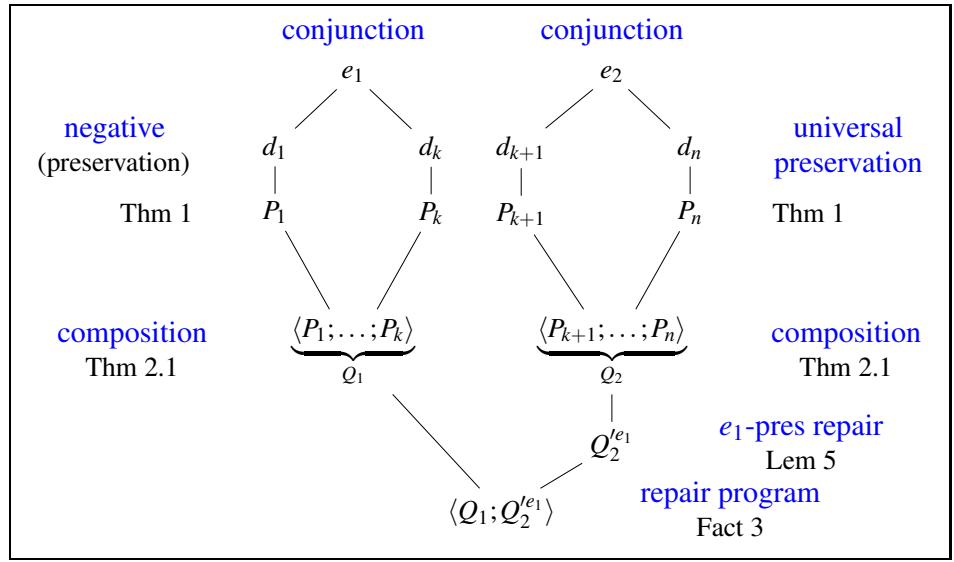
 
\[\scalebox{0.95}{\fbox{$\tikz[node distance=2em,shape=rectangle,outer sep=1pt,inner sep=2pt]{
\node(e1){$e_1$};
\node(d1)[strictly below left of=e1]{$d_1$};
\node(dk)[strictly below right of=e1]{$d_k$};
\node(P1)[strictly below of=d1,node distance=1em]{$P_1$};
\node(Pk)[strictly below of=dk,node distance=1em]{$P_k$};
\node(Q1)[node distance=8em,strictly below of=e1]{$\underbrace{\tuple{P_1;\ldots;P_k}}_{Q_1}$};
\node(T1)[strictly left of=P1]{Thm \ref{thm:repair}};
\node(T21)[strictly left of=Q1]{\begin{tabular}{c}{\blue\large composition}\\Thm \ref{thm:repair2}.1\end{tabular}};

\draw[-] (e1) -- (d1);\draw[-] (e1) -- (dk);
\draw[-] (d1) -- (P1);\draw[-] (dk) -- (Pk);
\draw[-] (P1) -- (Q1);\draw[-] (Pk) -- (Q1);
\node(e2)[node distance=8em,strictly right of=e1]{$e_2$};
\node(dk1)[strictly below left of=e2]{$d_{k+1}$};
\node(dn)[strictly below right of=e2]{$d_n$};
\node(Pk1)[strictly below of=dk1,node distance=1em]{$P_{k+1}$};
\node(Pn)[strictly below of=dn,node distance=1em]{$P_n$};
\node(Q2)[node distance=8em,strictly below of=e2]{$\underbrace{\tuple{P_{k+1};\ldots;P_n}}_{Q_2}$};
\node(d1aux)[node distance=1.5em,strictly left of=d1]{\begin{tabular}{c}\blue\large negative\\(preservation)\end{tabular}};
\node(dnaux)[node distance=1.5em,strictly right of=dn]{\blue\large \begin{tabular}{c}universal\\preservation\end{tabular}};
\node(e1aux)[node distance=0.2em,strictly above of=e1]{\blue\large conjunction};
\node(e2aux)[node distance=0.2em,strictly above of=e2]{\blue\large conjunction};
\node(Q2prime)[strictly below of=Q2,node distance=1em]{$Q_2^{\prime e_1}$};
\node(Q)[node distance=1em,strictly below left of=Q2prime]{$\tuple{Q_1;Q_2^{\prime e_1}}$};

\draw[-] (e2) -- (dk1);\draw[-] (e2) -- (dn);
\draw[-] (dk1) -- (Pk1);\draw[-] (dn) -- (Pn);
\draw[-] (Pk1) -- (Q2);\draw[-] (Pn) -- (Q2);
\draw[-] (Q2) -- (Q2prime);
\draw[-] (Q1) -- (Q);
\draw[-] (Q2prime) -- (Q);

\node(T1prime)[strictly right of=Pn]{Thm \ref{thm:repair}};
\node(T22)[strictly right of=Q2]{\begin{tabular}{c}{\blue\large composition}\\Thm \ref{thm:repair2}.1\end{tabular}};
\node(T3)[strictly right of=Q2prime]{\begin{tabular}{c}{\blue\large $e_1$-pres repair}\\Lem \ref{lem:repairing}\end{tabular}};
\node(T4)[strictly right of=Q]{\begin{tabular}{c}\blue\large repair program\\Fact \ref{fac:composition}\end{tabular}};
}$}}\]
\caption{\label{fig:illust}Illustration of the proof of Theorem \ref{thm:repair2}.3}
\end{figure}


In the following, we consider \emph{disjunctive} conditions, i.e. disjunctions of conditions. 
Whenever there exists a repair program for one of the subconditions, this repair program can be used for the disjunctive conditions as well. 
Every repair program for a condition is also a repair program for the corresponding disjunctive condition. 

\begin{theorem}[Repair III]\label{thm:disjunctive}

1. {If} $P$ is a repair program for $d$ and $d\impl d'$, {then} $P$ is a repair program for~$d'$. 
 
2. Every repair program for $d_1$ is repair program for $\bigvee_{i=1}^n d_i$.

3. {If} $P_1, \ldots, P_n$ are repair programs for $d_1, \ldots, d_n$, {then} $\{P_1, \ldots, P_n\}$ is a repair program for $\bigvee_{i=1}^n d_i$.
\end{theorem}

\ignore{
\begin{example} 
Consider the constraint\footnote{In the graphs, the TPArc is abbreviated with TP.} below, intuitively meaning, that double arcs from a transition to a place, or vice versa, are forbidden. (\cite{Radke18+a})  
\[\begin{array}{c}
d = \PA (\;\classlab{{\blue TP}}~\classlab{{\blue TP}}, \underbrace{\underbrace{\PE (\: \referenceedgelab{{\blue TP}}{Tr}{src} ~\referenceedgelab{{\blue TP}}{Tr}{src}\;)}_{c_1} \vee \underbrace{\PE \referenceedgelab{{\blue TP}}{Pl}{tgt} ~\referenceedgelab{{\blue TP}}{Pl}{tgt}}_{c_2}}_{c}\;) \wedge \\
\end{array}
\]

By the Repair Theorem \ref{thm:repair1}, there are repair programs for $P_{c_i} = \try \R_{c_i}$ for the proper conditions $c_i$ ($i \in \{1,2\}$), where $\R_{c_i}$ are the sets according to Construction~\ref{const:basic}. 

By Theorem \ref{thm:disjunctive}, the program $P_c = \{P_{c_1}, P_{c_2}\}$ is a repair program for $c$. 
Intuitively, it selects an occurrence of the two nodes of type TP, then it adds either a node of type Tr or a node of type Pl with the respective edge, or only the edge, afterwards it unselects. This is done as long as possible. 
\ignore{
The resulting repair program is:
\[\begin{array}{ll}
P_d = \tuple{\select(a, \neg c); P_c; \unselect(a)}\downarrow \mbox{ with } a \colon \emptyset \injto \classlab{{\blue TP}}~\classlab{{\blue TP}}\\
P_c = \{P_{c_1}, P_{c_2}\}\\
P_{c_1} = \try \R_{c_1}\\
P_{c_2} = \try \R_{c_2}\\
\end{array}
\] 
}
{\red little Problem: since the inner condition is a disjunction, it is not proper according to Theorem 1, since the disjunction is not proper}
{\red It would be better, if I could define it in a way, such that this can be handled by our approach.}
\end{example}
}

\begin{proof}
1. If $P$ is a repair program for $d$ and $d\impl d'$, then for every transformation $g\dder _{P} h$, $h\models d\impl d'$, i.e., $P$ is a repair program for $d'$. 
2. By $d_1\impl \bigvee_{i=1}^n d_i$ and statement 1, $P_1$ is a repair program for $\bigvee_{i=1}^n d_i$.
3. Since $d_i \impl \bigvee_{i=1}^n d_i$ and by statement 1 and 2,  $\{P_1, \ldots, P_n\}$ is a repair program for $\bigvee_{i=1}^n d_i$
\end{proof}

\ignore{
The repair program for disjunctive conditions is stable and terminating, provided that the repair program for the subcondition is stable and terminating.

\begin{lemma}[stability, termination]\label{lem:disjunctive-properties}

1. {If} $P_1, \ldots, P_n$ is stable, {then} $\{P_1, \ldots, P_n\}$ is stable.

2. {If} $P_1, \ldots, P_n$ is terminating, {then} $\{P_1, \ldots, P_n\}$ is terminating.
\end{lemma}

\begin{proof}
By stability and termination of $P_1, \ldots, P_n$, $\{P_1, \ldots, P_n\}$ is stable and terminating.
\end{proof}
}

There are repair programs for a large class of conditions: for all proper and generalized-proper ones, preserving conjunctions, and disjunctions.
In this context, a condition is said to be generalized proper, if it is obtained form a proper one by replacing a subcondition by a condition (over the same graph) with a repair program.

\begin{definition}[generalized proper]
Let $d = \Q(a,c)$ be a proper condition. A \emph{generalized proper} condition $d' = [c/c']$ is obtained by replacing $c$ with a condition $c'$, provided there exists a repair program for $c'$.
\end{definition}

\begin{definition}[legit conditions]
The class of \emph{legit} conditions is defined inductively as follows.
\begin{enumerate}
\item[(1)] If $d$ is proper or generalized proper, then $d$ is legit.
\item[(2)] If $d_1,\ldots,d_n$ are legit and preserving, then $\bigwedge_{i=1}^n d_i$ is legit.
\item[(3)] If $d_1$ is legit, then $\bigvee_{i=1}^n d_i$ is legit.
\end{enumerate}
\end{definition}

\begin{theorem}[Repair for legit conditions]For legit conditions, there is a repair program.\end{theorem}

\begin{proof}
By induction of the structure of conditions. Let $d$ be legit.

(1) If $d$ is proper, then, by Theorem \ref{thm:repair}, there is a repair program for $d$.

If $d$ is generalized proper, then $d$ is of the form $\Q(a,c)$ where $c$ is legit. By induction, there is a repair program for $c$. By a generalized Theorem~\ref{thm:repair}, there is a repair program for $d$.
%
(2) Let $ds=d_1,\ldots,d_n$ are legit and $ds$ preserving. By induction hypothesis, there are repair programs for $d_1,\ldots, d_n$. By Theorem \ref{thm:repair2}, there is a repair program for the conjunction $\bigwedge_{i=1}^n$.
(3) Let $d_1$ be legit. By induction hypothesis, there is a repair program for $d_1$. By Theorem \ref{thm:disjunctive}, there is a repair program for $\bigvee_{i=1}^n d_i$.
This completes the inductive proof.
\end{proof}

As a consequence of Construction \ref{const:proper}, we obtain the following.

\begin{lemma}[program properties]\label{lem:mpres}
The repair programs for proper conditions based on Construction~\ref{const:proper} are
(1) stable, (2) maximally preserving, (3) and terminating.
\end{lemma}

\begin{proof}
Let $d$ be a proper condition and~$P_d$ the program in Construction~\ref{const:proper}. 
(1) By the application condition $\ac = \Shift(A\injto B,\NE a)$, a rule in $\R_a$ can only be applied, iff the condition is not satisfied.
By the semantics of $\Skip, \try,$ and $\downarrow$, the repair programs are stable.
The proof of (2) and (3) can be found in \cite{Sandmann-Habel19a} and \cite{Habel-Sandmann18a}, respectively.
\end{proof}

\begin{lemma}[program properties]\label{lem:properties}
The repair programs for \legit conditions as above are stable \ignore{\red maximally preserving,}and terminating.
\end{lemma}

\begin{proof}
By induction on the structure of conditions.
Let $d$ be a legit condition and $P_d$ be the corresponding repair program. 

(1) If $d$ is proper, then by Lemma~\ref{lem:mpres}, the repair program $P_d$ is stable and terminating.
If $d$ is generalized proper, then, by induction hypothesis, $P_d$ is stable and terminating.

(2) By induction hypothesis, $P_i$, $P_i^{e_1}$ are stable\ignore{ maximally preserving,} and terminating, $\tuple{P_1; \ldots; P_n}$, and $Q_1, Q_2$ are stable and terminating. Finally $Q_2^{\prime e_1} = \tuple{P_{k{+}1}^{\prime e_1};\ldots;P_n^{\prime e_1}}$, with $P_j^{\prime e_1} = \tuple{P_j^{e_1}; P_{\NE a,j}^\id}$, is (a) stable and (b) terminating:
(a) By Construction, $P_{\NE a,j}^\id$ is only applicable, iff the condition is not satisfied. By the semantics of $\downarrow$, it is stable. 
(b) By Construction, $P_{\NE a,j}^\id$ is decreasing, consequently it is terminating. 
Consequently, $\tuple{Q_1;Q_2^{\prime e_1}}$ is stable and terminating.

(3) If $d = \bigvee_{i}^n d_i$, then, by induction hypothesis, $P_{1}$ is stable and terminating. Consequently, $\{P_1, \ldots, P_n\}$ is stable and terminating.
\end{proof}


\begin{remark}[Implementation]
The approach to graph repair has been implemented in \textsc{ENFORCe$+$}.
\end{remark}
\section{Application to meta-modeling}\label{sec:application}

The standard tool for model-driven engineering is the Eclipse Modeling Framework (EMF). 
In \cite{Biermann-Ermel-Taetzer12a}, an EMF model graph is defined as a typed graph, representing the model, satisfying the following conditions: No node has more than one container.
There are no two parallel edges of the same type.
No cycles of containment occur.
For all edges in the opposite edges relation, there exists an edge in opposite direction.   

\begin{definition}[EMF-model graph]\label{def:emf}
A typed graph $G$~is an \emph{\mbox{EMF}-model graph}, if it satisfies the following conditions:
\[\begin{tabular}{p{4cm}p{8.5cm}}
1. At most one container &$\forall e_1, e_2 \in C_G$.  $\tar_G(e_1) = \tar_G(e_2)$ implies $e_1 = e_2$\\
2. No containment cycle &$\forall v\in V_G$.  $(v,v) \not\in\contains_G$\\
3. No parallel edges &$\forall e_1, e_2 \in E_G$. $\src_G(e_1) = \src_G(e_2)$, $\tgt_G(e_1) = \tgt_G(e_2)$, and $type_{E_G}(e_1) = type_{E_G}(e_2)$ implies $e_1 = e_2$\\
4. All opposite edges& $\forall (e_1,e_2) \in O$. $\forall e'_1\in E_G. type_G(e'_1) = e_1$. 
$\exists e'_2\in E_G$. $type_G(e'_2) = e_2, \src_G(e'_1) = \tgt_G(e'_2)$ and $\src_G(e'_2) = \tgt_G(e'_1)$\\
\end{tabular}\]
The set $C_G$\ignore{ = $\{e\in E_G\mid type_G(e) \in C\}$} denotes the set of edges in $G$ which are typed by a containment edge.
$\contains_{G} \subseteq V_G \times V_G$ is the \emph{containment relation} induced by the set $C\subseteq E_T$:
If $e\in C$ and $v_1\leq\sou(e),v_2 \leq\tar(e)$, then $(v_1,v_2)\in\contains_{G}$. If  $(v_1,v_2) ,(v_2,v_3) \in \contains_{G}$, then $(v_1,v_3)\in\contains_{G}$.\\
The conditions are said to be \emph{EMF-constraints}.
\end{definition}

The second constraint is a monadic second order constraint. Instead of it, we consider the constraint ``No containment cycle of length $\leq k$'' for a fixed natural number~$k$. The resulting constraints, called \emph{\EMF constraints}, are first-order constraints and can be expressed by typed graph constraints \cite{Habel-Pennemann09a}.


\begin{fact}[\EMF-constraints]\label{def:EMFnested} 
For the \EMF-constraints, there is a schema of\ignore{ corresponding} typed graph constraints:
\[\begin{tabular}{p{4cm}p{1.5cm}p{5cm}}
1. At most one container &&$\NE \atmostonecontainer$\\
2. No containment cycle \phantom{2.\;\;}of length $\leq k$&&$\NE\containmentloop \wedge \bigwedge_{i = 1}^{k{-}1} \NE\containmentcycle{i}$\\
3. No parallel edges&&$\NE\paredgelab{A}{B}{t}{t}$\\
4. All opposite edges&& $\PA(\directededgelabO{A}{B}{t_1}{t_2}, \PE \oppositeedgecondlabO{A}{B}{t_1}{t_2})$\\
\end{tabular}\]
where A,B,C,t,t$_1$,t$_2$ are node and edge types, respectively, edges without type are arbitrary typed, and $\scalebox{1}{\containmentp{i}}$ denotes a path of containment edges of length~$i$.
\end{fact}

The first constraint requires that there are no two different containment edges with a common target. The second constraint requires that there are no loops and no cycles of length $\leq k$. The third constraint requires, that there are no parallel edges of the same type. 
The fourth constraint requires that, if there is an opposite-edge marking between an A-typed and a B-typed node  with type requirement t$_1$,t$_2$, there exists already one edge $e_1$ with the type t$_1$, then an opposite edge with type t$_2$ in opposite direction should exist. 

\begin{fact}\label{fac:trivial}
The instances of \EMF constraints are negative or universal.
Every conjunction of \EMF constraint instances is satisfiable.
Every sequence of instances of one \EMF constraint is preserving.
\end{fact}

\begin{lemma}[preservation\ignore{of \EMF constraints}]\label{lem:independentEMF}
Every conjunction of instances of an \EMF constraint is preserving.
\end{lemma}

\begin{proof} The instances of the first three \EMF constraints are negative; thus, each conjunction of them is preserving. 
The instances of the forth \EMF constraint are universal; by induction on the number of constraints, it can be shown that each conjunction is preserving.
\longv{
\begin{notation}
Let $d_i$ be an instance of the opposite-edge constraint. 
\[d_i=\PA(\;\directededgelabO{A_i}{B_i}{t_{1i}}{t_{2i}}\;,\;\PE \oppositeedgecondlabO{A_i}{B_i}{t_{1i}}{t_{2i}\;\;}\;)\]
The opposite-edge requirement in~$d_i$ is abbreviated by $\tuple{A_i,\Brm_i,\type_{1i},\type_{2i}}$. 
The edges in $d_i$ are called $\tuple{\Arm_i,\Brm_i,\type_{i1}}$ and $\tuple{\Brm_i,\Arm_i,\type_{i2}}$-edges, respectively. 
The repair program for $d_i$ is as follows: 
\[P_i=\tuple{\select(a,\neg c);
\directededgelabO{A_i}{B_i}{\type_{1i}}{\type_{2i}}\dder\oppositeedgecondlabO{A_i}{B_i}{t_{1i}}{t_{2i}}\;,\;\NE\oppositeedgecondlabO{A_i}{B_i}{\type_{1i}}{\type_{2i}}\;;\;\unselect(a)}\downarrow.\] 
\end{notation}

$n = 1$. Let $G\dder_{P_1} H$ with $G\models d_0=\ctrue$. Then $P_1$ is  $\ctrue$-preserving.

$n\to n+1$. Let $G\dder_{\tuple{P_1;\ldots;P_{n+1}}}H$ be an arbitrary transformation. Then the transformation is of the form $G\dder_{\tuple{P_1;\ldots;P_n}}H\dder_{P_{n+1}} M$.
By induction hypothesis, $P_1,\ldots,P_n$ is $d_0, \ldots, d_n$-preserving, i.e. for $k=1,\ldots,n$, $P_k$ is $\wedge_{i={\red 0}}^n d_i$-preserving. 

Consider now the transformation step $H\dder_{P_{n+1}} M$. Then $M$ consists of edges from $H$ as well as edges created by the program $P_{n+1}$. 
The first ones satisfy the opposite-edge requirement; for the second ones we have to show it. If we can show that there is no {\red proper application} of the programs $P_1,\ldots, P_n$ to $M$, then the graph $M\models\wedge_{i=1}^n d_i$, and, since $P_{n+1}$ is a repair program for $d_{n+1}$, $M\models d_{n+1}$. This holds for arbitrary transformations. Thus, 
$P_1,\ldots,P_{n+1}$ is $d_0, \ldots, d_{n+1}$-preserving, i.e. for $k=1,\ldots,n+1$, $P_k$ is $\wedge_{i=1}^{n+1} d_i$-preserving. In the following, let $o:=n+1$. 

\[\begin{tikzpicture}[node distance=2.5em]
\node[node] (A) {\small $\Arm_o$};
\node[node,strictly right of=A] (B) {\small $\Brm_o$};
\draw[<->] (A) to node[above left] {\footnotesize $\type_{o1}$} node[below right] {\footnotesize $\type_{o2}$} (B);
\draw[->] (A) to [bend left=40] node[above] {\footnotesize $\type_{o1}$} (B);
\draw[->,blue] (B) to [bend left=40] node[below] {\footnotesize  $\type_{o2}$} (A);
\node[node,node distance=3em,strictly below of=A] (A') {\small $\Brm_m$};
\node[node,strictly right of=A'] (B') {\small $\Arm_m$};
\draw[<->] (A') to node[above left] {\footnotesize $\type_{m2}$} node[below right] {\footnotesize $\type_{m1}$} (B');
\draw[->,blue] (B') to [bend left=40] node[below] {\footnotesize  $\type_{m1}$} (A');
\draw[->,dashed] (A') to (A);\draw[->,dashed] (B') to (B);
\end{tikzpicture}\]

 Consider a $\tuple{\Brm_o,\Arm_o,\type_{\red o2}}$-edge created in the last transformation step. 
 Then there exist a $\tuple{\Arm_o,\Brm_o,\type_{o1}}$-edge in $M$ with opposite-edge requirement  $\tuple{\Arm_o,\Brm_o,\type_{o1},\type_{\red o2}}$. 
 If there is a match for the plain rule of the program $P_m$ ($m<o$) using the created edge, then a $\Arm_m=\Brm_o$, $\Brm_m=\Arm_o$, $\type_{m1}=\type_{o2}$. 
By the definition of type graphs (Definition~\ref{def:type-graph}), the relation~$O$ is functional, i.e., \ignore{$ \PA((e_1,e_2),(e_1,e3)\in O$, $e2 = e3$,} there is only one opposite-edge requirement between the nodes in consideration. Thus, $\type_{m2}=\type_{\red n1}$. \CS{I think, $m2 = o1$ } 
Thus, there is an opposite-edge for the edge in consideration, i.e., the application condition of the rule in $P_m$ is not satisfied and $P_m$ cannot applied. 
Thus, $P_1.\ldots,P_{n+1}$ is $d_1,\ldots,d_{n+1}$-preserving. This completes the inductive proof.
}
\end{proof}

\ignore{
\begin{proof}
By induction on the number $n$ of pure constraints of opposite edges.
{Let $d = d_1, \ldots, d_n$ be the sequence of opposite edge constraints, and $P_i^\prime$ the $c$-preserving repair program (see Theorem \ref{thm:repair2}). By Construction, for each $d_i$, the program $P_i'$ either:
\begin{itemize}
\item[(1)] adds exactly one edge of type $t_{2,i}$, provided it preserves the negations, or
\item[(2)] deletes exactly one edge of type $t_{1,i}$, otherwise,
\end{itemize}
between a node of type $A_i'$, and $B_i'$ (see Figure \ref{fig:proof-illustration}).
\begin{figure}[h]
\[\begin{tikzpicture}[node distance=4em]
\node (TG) {$TG:$};
\node[node,strictly right of=TG] (A) {$\mathrm{A_i}$};
\node[node,strictly right of=A,node distance=4em] (B) {$\mathrm{B_i}$};
\draw[opposite] (A) to node[above right](t1) {$t_{1,i}$} node[below left](t2) {$t_{2,i}$} (B);
\node[strictly below of=TG] (G) {$G:$};
\node[node,strictly below of=A] (A') {$\mathrm{A'_i}$};
\node[node,strictly below of=B] (B') {$\mathrm{B'_i}$};
\draw[opposite] (A') to node[above right] {$t_{1,i}$} node[below left] {$t_{2,i}$} (B');
\draw[->] (A') to [bend left=40] node[above right](t1') {\footnotesize $t_{1,i}$} (B');
\draw[->,bend left=40] (B') to node[below left](t2') {$t_{2,i}$} (A');
\draw[->,dashed] (A') to (A);
\draw[->,dashed] (B') to (B);
\draw[->,dashed] (t1') to (t1);
\draw[->,dashed,bend left] (t2') to (t2);
\end{tikzpicture}\]
\caption{\label{fig:proof-illustration} For each opposite edge pair, the typing morphism $type_i$ is uniquely defined}
\end{figure}

For $n = 2$, let $d = d_1 \wedge d_2$. 
By Construction, $P_1$ adds an edge of type $t_{2,1}$, or deletes an edge of type $t_{1,1}$. In both cases, $P_1$ is a repair program for $d_1$.
Analogously, $P_2$ adds an edge of type $t_{2,2}$, or deletes an edge of type $t_{1,2}$.
Since each added edge of $P_2$ has either a different source, target, or edge type, than the edges of $P_1$, the program $P_2$ is independent from $d_1$.
\[\begin{tabular}{l|lllll}
& $P_{1,1}$ & $P_{1,2}$ & $P_{2,1}$ & $P_{2,2}$ &\\ \hline 
added edge& $t_{2,1}$ & $t_{1,1}$ & $t_{2,2}$ & $t_{1,2}$ \\ \hline 
deleted edge& $t_{1,1}$ & $t_{2,1}$ & $t_{1,2}$ & $t_{2,2}$ \\ 
\end{tabular}\]

Let $n= n+1$ and $d = d_1, \ldots, d_{n+1}$. By Induction Hypothesis $d_1, \ldots, d_n$ are independent. By Construction, $P_{n+1}$ adds an edge of type $t_{2,n+1}$, or deletes an edge of type $t_{1,n+1}$. Assume, that, the added edge, is of the same type, as one of the edges before, with the same source, and target node. Then, there exists a parallel edge of the same type between this source, and target node, which is a contradiction, to the constraint No parallel edges.}
\end{proof}
}

Let $\emfk_1,\emfk_2$ be conjunctions of \EMF constraints. An \emph{\EMF model repair} program for 

$\tuple{\emfk_1,\emfk_2}$ is an $\emfk_1$-preserving repair program for $\emfk_2$. 
An \emph{\EMF model completion} program is an EMF model repair program for $\ctrue$ and for the conjunction of all \EMF constraints.
A repair program $P$ for a constraint $e$ is \emph{stable} if, for all transformations $L\dder_P M$, $L\models e$ implies $L\cong M$, i.e., they do not change the the graph provided the constraint $e$ is satisfied.

\begin{theorem}[\EMF model repair \& completion]\label{thm:emfk-repair}
Let $\emfk_1$ is  a conjunction of negative \EMF constraints, $\emfk_2$ a conjunction of negative or universal and preserving \EMF constraints, and $\bar{\emfk}\; $the conjunction of all \EMF constraints.
\begin{enumerate}
\item There is a model-repair program for $\tuple{\emfk_1,\emfk_2}$.
\item There is a model-completion program. 
\item The \EMF model repair and completion programs are stable.
\end{enumerate}
\end{theorem}

\begin{figure}[htp]
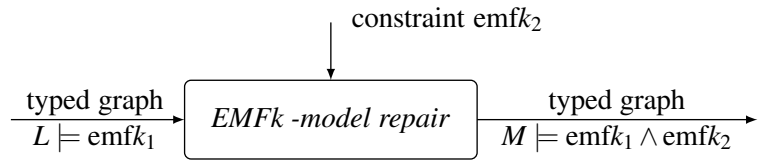

\[\tikz[node distance=1.2em,label distance=0pt,outer sep=0pt,inner sep=1pt]{
\node(MM)[label={[marking]right:\begin{tabular}{c}typed graph\\ $L\models \emfk_1$\end{tabular}}]{};
\node[rectangle,inner sep=10pt,round corners](Trans)[label={[marking,outer sep=0pt,inner sep=0pt]right:\begin{tabular}{@{}c@{}}\end{tabular}},node distance=6em,strictly right of=MM]{\emph{\EMF-model repair}};
\node(GG)[node distance=3em,strictly right of=Trans]{};
\node(rule)[node distance=2em,strictly above of=Trans,label={[marking,outer sep=0pt,inner sep=1pt]right:\begin{tabular}{c}{constraint} $\emfk_2$\end{tabular}}]{};
\node(Generate)[node distance=6.5em,strictly right of=GG]{};
\draw[arrow] (MM) -- (Trans);
\draw[arrow] (rule) -- (Trans);
\draw[arrow] (Trans) -- node{\begin{tabular}{c}typed graph\\ $M\models \emfk_1\wedge \emfk_2$ \\\end{tabular}}(Generate);}\]
\caption{\label{fig:illustrationtheoremmodelrepair} Illustration of \EMF model repair }
\end{figure}

\begin{proof}
1. By Theorem \ref{thm:repair}, there are repair program $Q_1,Q_2$ for $\emfk_1,\emfk_2$, respectively.
By Lemma~\ref{lem:independentEMF}, every conjunction $\emfk_1$ and $\emfk_2$ is preserving.
Consequently, \ignore{By the special assumption for $e_1,e_2$, ]}Lemma \ref{lem:presrepair} can be applied and there is an $\emfk_1$-preserving repair program $Q_2^{\prime ~\emfk_1}$ for $\emfk_2$, and by Fact \ref{fac:composition}, $\tuple{Q_1,Q_2^{\prime ~\emfk_1}}$ is a repair program for $\emfk_1\wedge \emfk_2$.

2. The statement is an immediate consequence of Theorem \ref{thm:emfk-repair}.1.

3. The statement follows immediately from Construction \ref{const:proper}.
\end{proof}

\begin{remark}[(OCL-)Constraints]
By the repair results on typed graphs, (model) repair and completion can be done for other constraints satisfying the requirements in Theorem~\ref{thm:repair2}, e.g., for first-order \\(OCL-)constraints.
\end{remark}


Inspecting the \EMF repair (completion) program, it turns out that the program deletes and adds an edge, but it does not change the number of nodes.

\begin{fact}[preservation of the number of nodes]\label{fac:nodes}
The application of the \EMF repair (completion) program does not change the number of nodes.
\end{fact}

There is a close relationship between \EMF and EMF. For an \EMF constraint $\emfk$, $\emf$ denotes the more rigorous EMF constraint requiring no containment cycles and, for an EMF-constraint $\emf$, $\emfk$ denotes the weaker \EMF constraint $\emf$ requiring no containment cycles of length $\leq k$.

\begin{fact}[\EMF-EMF]\label{fac:emf}For typed graphs $L$ of node size $\leq k$, $L\models\emfk$ iff $L\models\emf$. \end{fact}

As a consequence, we obtain the following statement for EMF model repair \& completion.

Let $\emf_1,\emf_2$ be conjunctions of EMF constraints.
There is \emph{EMF model repair} for a typed graph $L\models\emf_1$ and $\emf_2$ if there is a program $P$ such that, for all transformations $L\dder_P M$, $M \models \emf_1 \wedge \emf_2$.
There is an \emph{EMF model completion} for a typed graph $L$ if there is a program $P$ such that, for all transformations $L\dder_P M$, $M$ is an EMF-model graph.
Model repair and completion for a constraint $e$ are \emph{stable} if, for all typed graphs $L\models e$, all repairs (completions) yield a typed graph isomorphic to $L$.

\begin{theorem}[EMF model repair \& completion]\label{thm:emf-repair}
Let $\emf_1$ is  a conjunction of negative EMF-constraints, $\emf_2$ a conjunction of negative or universal and preserving EMF-constraints, and $\bar{\emfk}$ the conjunction of all \EMF constraints.
\begin{enumerate}
\item There is a EMF model repair for all typed graphs $L\models\emf_1$ and $\emfk_2$.
\item There is a EMF model completion for all typed graphs $L$.
\item The EMF model repair and completion for a constraint $e$ are stable for all typed graphs $L$.
\end{enumerate}
\end{theorem}

\begin{figure}[htp]
\[\scalebox{0.9}{
\tikz[node distance=6em,inner sep=1pt,outer sep=2pt]{
\node(L)[label={above right:}] {$\emf_1 \vmod ~L$};
\node(G)[below of=L]{$\emfk_1 \vmod ~L$};
\node(H)[right of=G,node distance=12em]{\begin{tabular}{c}$M\models \emfk_1 \wedge \emfk_2$\end{tabular}};
\node(M')[right of=L,node distance=12em] {\begin{tabular}{c}$M\models \emf_1 \wedge \emf_2$\end{tabular}};

\draw[derivation] (L) to node[left] {\begin{tabular}{c}if $|V_L| = k$ \end{tabular}} (G);
\draw[derivation] (L) to node[above] {$P$} (M');
\draw[derivation] (G) to node[above] {$P$} (H);
\draw[derivation] (H) to node[right] {\begin{tabular}{c}if $|V_M| = k$ \end{tabular}} (M');
}}\]
\caption{\label{fig:emf-emfk}Relation on $\emf$ and $\emfk$}
\end{figure}

\begin{proof} 1. For a typed graph $L$ of node size $k$ satisfying $\emf_1$ and $\emf_2$, we take the \EMF repair program $P$ for $\tuple{\emfk_1,\emfk_2}$ and apply it to $L$. 
By Fact \ref{fac:emf}, $L\models\emf_1$ implies  $L\models\emfk_1$.
By Theorem~\ref{thm:repair2}, the application of $P$ to $L$ yields a typed graph $M$ satisfying $\emfk_1\wedge\emfk_2$.
The program does not change the number of nodes, i.e., $M$ is a typed graph with $k$ nodes. 
By Fact \ref{fac:emf}, $M$ satisfies $\emf_1\wedge\emf_2$.

2. For a typed graph $L$ of node size $k$, we take the \EMF-completion program and apply it to $L$ yielding an \EMF-model graph $M$, i.e, a typed graph $M$ satisfying $\bar{\emfk}$. The program does not change the number of nodes, i.e., $M$ is a typed graph with $k$ nodes. 
By Fact \ref{fac:emf}, $M$ satisfies $\emf_1 \wedge \emf_2$.

3. By Theorem \ref{thm:emfk-repair}, the \EMF-model repair and completion programs are stable, i.e.,
for all transformations $L\dder_P M$, $L\models e$ implies $L\cong M$. Applying the programs to a typed graph, the property remains preserved.
\end{proof}

\begin{remark}[Repair of other structures]
In this section, the results in Section \ref{sec:repair} are applied to meta-modeling: typed graphs are repair w.r.t. EMF constraints. Obviously, typed graphs can also repaired w.r.t. other constraints, e.g. OCL-graph constraints as considered in \cite{Radke+18a}.
Note that the presented results hold in every $\M$-adhesive category with $\Epi'$-$\M$-pair factorization. 
As a consequence, we can do repair for 
high-level structures and high-level constraints \cite{Ehrig-Golas-Habel-Lambers-Orejas14a}.\end{remark}

\begin{remark}[Model generation] 
In model generation, given a meta-model, one tries to find some (all) instances of the meta-model. 
Model generation may be seen as a special case of model completion applying the program: 
{For a fixed $k$, the application of the \EMF model completion program to the empty typed graph yields \EMF model graphs. 
By Fact~\ref{fac:nodes}, every \EMF model graph with node size $\leq k$ is an EMF model graph. }
In this way, we obtain some instances of the meta-model. 
\end{remark}

\ignore{ 
\begin{corollary}[EMF model completion]
For every typed graphs $G$, there is a completion program $P_d$ such that for all transformations $G \dder_{P_d} H$, $H$ is an EMF model graph. 
\end{corollary}
\begin{proof}
Let $k$ be the number of nodes in $G$, $d$ the conjunction of all \EMF constraints, and $P_d$ the corresponding \EMF completion program. Then, for every transformation $G\dder_{P_d} H$, $H$ is an \EMF model graph. All cycles in $G$ are of length $\leq k$, and the \EMF model repair program does not change the number of nodes, i.e., $|V_G| = k = |V_{H}|$. Since typed graphs of node size~$k$ only may possess cycles of length~$k$, $H$ is an EMF model graph, i.e. $H \models$ EMF constraints $\Leftrightarrow H \models$ \EMF constraints.
\end{proof}
}
%

\section{Related work}\label{sec:related}
In this section, we present some related concepts on model repair, for which there is a wide variety of different approaches. Recently, there has been a sophisticated survey on different model repair techniques, and a feature-based classfication of these approaches, see \cite{Macedo-etal17a}.
In their sense, our approach is \emph{stable} (see Theorem~\ref{thm:emfk-repair}.3).


\ignore{ 
In \textbf{Nentwich et al.} \cite{Nentwich-etal03a} a repair framework for inconsistent distributed UML documents is presented.
Given a first order formula, the algorithm automatically creates a set of repair actions from which the user can choose, when an inconsistency occurs.  
These repair actions can either delete, add or change a model document. 
It can be shown, that the repair actions are correct and complete. 
The problem of repair cycles is left for future work. 
Since, in general, it is undecidable, if a constraint is satisfiable, the algorithm may not terminate.
}

In {\bf Schneider et al. 2019} \cite{Schneider-etal19a}, a logic-based incremental approach to graph repair is presented, generating a sound and complete (upon termination) overview of least changing repairs. The graph repair algorithm takes a graph and a first-order (FO) graph constraint as inputs and returns a set of graph repairs. Given a condition and a graph, they compute a set of symbolic models, which cover the semantics of a graph condition.
Both approaches are proven to be \textbf{correct}, i.e. the repair (programs) yield to a graph satisfying the condition. 
The delta-based repair algorithm takes the graph update history explicitly into account, i.e. the approach is \textbf{dynamic}. In contrast, our approach is static, i.e., we first construct a repair program, then apply this program to an arbitrary graph. 
The repair algorithm does not \textbf{terminate}, if the repair updates trigger each other ad infinitum. Here, we have constructed terminating repair programs.

In \textbf{Biermann et al. 2012} \cite{Biermann-Ermel-Taetzer12a}, for EMF model transformations, consistent transformations are defined. For a set of rules, they slightly modify them, to get so-called consistent transformation rules. This way, a direct transformation step applied at an EMF model yields an EMF model graph again. 
In our approach, a direct transformation step leads to typed graphs. We use the repair program to complete the typed graph to an EMF model graph.

In {\bf Nassar et al. 2017} \cite{Nebras-etal17a}, a rule-based approach to support the modeler in automatically trimming and completing EMF models is presented. For that, repair rules are automatically generated from multiplicity \ignore{(mult)}constraints imposed by a given meta-model. The rule schemes are carefully designed to preserve the EMF model constraints. 
One can use the approach in this paper to transform a typed graph to an \EMF model graph, then, one can use the approach of \cite{Nebras-etal17a}, to transform an EMF model graph to an EMF model graph, satisfying additional multiplicity constraints.
\longv{
\[\scalebox{0.9}{
\tikz[node distance=1.2em,label distance=0pt,outer sep=0pt,inner sep=1pt]{
\node(MM)[label={[marking]right:\begin{tabular}{c}EMF model\\graph\end{tabular}}]{};
\node[rectangle,inner sep=10pt,round corners](Trans)[label={[marking,outer sep=0pt,inner sep=0pt]right:\begin{tabular}{@{}c@{}}\end{tabular}},node distance=6em,strictly right of=MM]{\large\begin{tabular}{c}trimming \& completion\end{tabular}};
\node(GG)[node distance=3em,strictly right of=Trans]{};
\node(Generate)[node distance=6em,strictly right of=GG]{};
\draw[arrow] (MM) -- (Trans);
\draw[arrow] (Trans) -- node{\begin{tabular}{c}EMF model\\graph + multiplicities\end{tabular}}(Generate);
}}\]}

In \textbf{Nassar et al. 2020} \cite{Nassar-etal20a}, a method to simplify constraint-preserving application conditions is presented. 
Their simplifications of the application conditions are based on three main concepts: 
(1) If the elements which are deleted (or added) by a rule are type-disjoint with the types of the constraint, i.e. they share no types, the application condition simplifies to $\ctrue$,
(2) For increasing (or decreasing) rules, and positive (or negative) constraints, the application is $\ctrue$,
(3) For negative constraints $\NE C$ one may omit the cases where $C$ and the elements created by the rule overlap in at least one element. 
The simplified application conditions are proven to be logically equivalent to the original application condition.
The results are proven to be correct for $\M$-adhesive categories, and can be used to simplify the application conditions needed for our construction of condition-preserving application conditions (Lemma~\ref{lem:pres}), as well. 
Furthermore, the EMF model constraints (Fact~\ref{def:EMFnested}) are simplified by replacing every subcondition violating the no parallel edge, and at most one container constraints with $\cfalse$.

In \textbf{Kosiol et al. 2020} \cite{Kosiol-etal20a} two notions of consistency as a graduated property are introduced: consistency-sustaining rules do not change the number of violations of a constraint in a graph, and consistency-improving rule reduce the number of violations in a graph. 
The definition is based on the so-called consistency index, given by the number of constraint violations in a graph, divided by the number of ``relevant occurrences'' of the constraints in a graph. 
A transformation is consistency sustaining, if the consistency index for the input graph is equal or less than the resulting graph, and consistency improving if the number of the violations in the resulting graph is smaller than in the input graph. 
A rule is consistency sustaining, if all transformations are. 
A rule is consistency improving, if all applications of the rule are consistency sustaining, there exists a graph with constraint violations, and a transformation, such that the number of the violations in the resulting graph is smaller than in the input graph. 
A rule is strongly consistency improving if all its applications to a graph with constraint violations is consistency improving.
In their setting, the rules derived from our approach are strongly consistency improving,

In {\bf Taentzer et al. 2017} \cite{Taentzer-etal17a}, a designer can specify a set of so-called change-preserving rules, and a set of edit rules.
Each edit rule, which yields to an inconsistency, is then repaired by a set of repair rules.
The construction of the repair rules is based on the complement construction.
It is shown, that a consistent graph is obtained by the repair program, provided that each repair step is sequentially independent from each following edit step, and each edit step can be repaired. 
The repaired models are not necessarily as close as possible to the original model.

In \textbf{Rabbi et al. 2015} \cite{Rabbi-etal15a}, a model completion approach for predicates specified in the Diagrammic Predicate Framework (DPF) is introduced. For every predicate in the model, they derive a set of completion rules, by constructing the pullback of the instance, the meta-model and the graph of the condition. These rules, applied as long as possible, yields a model which conforms to the predicate. 
In our approach, the rules are derived from the constraint and the meta-model. In both approaches, the meta-model remains unchanged. 

In \textbf{Wang 2016} \cite{Wang16a}, the semantics of the predicates in the DPF is specified as graph constraints, and a model repair approach for these graph constraints is introduced. For constraints of the form $\PA (L, \PE R)$ or $\PA(L, \NE R)$, repair rules are directly derived from the constraints. The construction is based on the construction of subgraphs of $L$ and $R$. For the constraint $\PA (L, \PE R)$, for each subgraph~$B$, they derive rules $\tuple{B \dder R, \NE R}$ and $\tuple{L \dder B, \NE R}$. For the constraint $\PA(L, \NE R)$, rules of the form $\tuple{L \dder B}$ are derived. The performance of the approach has been optimized for practical application scenarios. In this work, we have combined the programs for proper conditions to a repair program for conjunctions of proper conditions. The properties of the repaired conditions remain preserved, whenever possible. If this is not possible, we delete the occurrence. As far as we can see, the approach in \cite{Wang16a} does not handle conjunctions.

\ignore{
In \textbf{Puissant et al.} \cite{Mens-etal15a}, a regression planner is used to automatically generate sequences of repair actions that transform a model with inconsistencies to a valid model. The initial state of the planner is the invalid model, represented as logical formula, the accepting state is a condition specifying the absence of inconsistencies. Then, a recursive best-first search is used to find the best suitable plan for resolving the inconsistencies. The correctness of the algorithm is not proven, but the approach is evaluated through tests on different UML models. {\red In contrast to our approach, the algorithm takes the decision of the user into account and uses backtracking, to find ``the best'' suitable plan to repair the inconsistencies. In contrast, our approach is purely automatic.}
}

In \textbf{Barriga et al 2019} \cite{Barriga-etal19a}, an algorithm for model repair based on EMF is presented, which relies on reinforcement learning. For each error in the model, a so-called Q-table is constructed, storing a weight for each error, \ignore{\red context}and repair action. This weight indicates how good a repair action is, depending on the repair action and regarding the users preferences. 
The approach can repair errors provided by the EMF diagnostician. The results are not proven but evaluated using mutation testing.

\section{Conclusion}\label{sec:conclusion}
In this paper, we have presented the theory of typed repair programs, applied it to \EMF-and EMF model graph repair. 
\begin{enumerate}
\item {\bf Typed graph repair.} We have extended our results on graph repair to typed graphs. 
There are repair programs for a large class of conditions, called \emph{legit} conditions.
Application of the repair programs to an arbitrary typed graph yields a typed graph satisfying~the condition.
\item {\bf \EMF model repair.} For \EMF constraints, a first-order variant of EMF constraints,
we present stable \EMF model repair and completion programs. 
Application of these programs to any typed graph yields a repaired typed graph and an \EMF model graph, respectively,
\item \textbf{EMF model repair.} These results are applied to the EMF world and yield to EMF model repair and completion results.
\end{enumerate}

\ignore{ 
In the present approach, the repair is done by deletion and inserting of items: If there are two containment edges into one container, this is repaired by deletion of one of the edges. In the rule-based approach \cite{Sandmann-Habel19a}, the repair program is constructed from a given set of rules. In the case of the constraint ``At most one container'', a containment edge could be redirected provided that a redirection rule is in the rule set. Similar, containment cycles could be {\red repaired} by redirecting one containment edge. 
\ignore{\red One possible way, would be a repairing set, which merges nodes, instead of deleting them. If this can be proven to be compatible with the repairing sets here, we could merge nodes, instead of deleting them.}
}

Further topics may be the following. 

1. {\bf Least changng repairs.} Our repair programs induce \emph{repairs} in the sense of Schneider et al. \cite{Schneider-etal19a}. It would be nice to show that these induced repairs are least changing repairs. 

2. {\bf Redirection of edges.} Our repair programs try to preserve items; if this is not possible, they delete items. In  Nassar et al. \cite{Nebras-etal17a}, multiplicity constraints are considered. In this context, they use the idea, to redirect an edge instead of deleting it. How this can be included in our approach?

3. {\bf Generalization to attributed type graphs} as e.g. in \cite{Radke+18a,Orejas-Lambers10a}. 


\paragraph{\normalfont\textbf{Acknowledgements.} We are grateful to Annegret Habel, Marius Hubatschek, Jens Kosiol, Okan Özkan, Gabriele Taentzer, and the anonymous reviewers for their helpful comments to this paper.}

\bibliographystyle{eptcs}
\bibliography{bib/lit-GRAGRA,bib/lit-LOGIK,bib/lit-model-repair,bib/lit-MCheck,bib/lit-ocl,bib/lit-PN}

\newpage
\begin{appendix}
\end{appendix}
\end{document}